\documentclass[a4paper]{article} 

\usepackage{amsmath}
\usepackage{amssymb}
\usepackage{mathrsfs}
\usepackage{amsfonts}
\usepackage{amsthm}
\usepackage{paralist}
\usepackage{graphicx}
\usepackage{tikz-cd}
\usepackage{multicol}

\newtheorem{theorem}{Theorem}[section]
\newtheorem*{theorem*}{Theorem}
\newtheorem*{problem*}{Question}
\newtheorem*{CRT}{Chinese Remainder Theorem}
\newtheorem{lemma}[theorem]{Lemma}
\newtheorem{proposition}[theorem]{Proposition}
\newtheorem{corollary}[theorem]{Corollary}
\theoremstyle{definition}

\theoremstyle{nonumberplain}
\newtheorem{claim}{Claim}

\newcommand{\Con}{\mathop{\mathrm{Con}}}

\newcommand{\Conn}{\mathop{\mathbf{Con}}}
\newcommand{\Clo}{\mathop{\mathrm{Clo}}}
\newcommand{\dom}{\mathop{\mathrm{dom}}\nolimits}
\newcommand{\var}{\mathop{\mathrm{var}}}
\newcommand{\varsets}{\mathop{\mathrm{varsets}}}
\DeclareMathOperator{\con}{\mathrm{Con}}

\newcommand{\prob}[1]{\textup{\textsf{#1}}}
\newcommand{\path}{\rightsquigarrow}
\newcommand{\cross}[4]{{{{#2} \, {#4}} \choose {{#1} \, {#3}}}}
\renewcommand{\mod}[1]{\mathrel{(\text{\rm mod } #1)}}

\begin{document}

\title{The complexity of the Chinese Remainder Theorem}

\author{Miguel Campercholi, Diego Castaño, Gonzalo Zigarán}

\maketitle

\begin{abstract}
The Chinese Remainder Theorem for the integers says that every system of congruence equations is solvable as long as the system satisfies an obvious necessary condition. This statement can be generalized in a natural way to arbitrary algebraic structures using the language of Universal Algebra. In this context, an algebra is a structure of a first-order language with no relation symbols, and a congruence on an algebra is an equivalence relation on its base set compatible with its fundamental operations. A tuple of congruences of an algebra is called a Chinese Remainder tuple if every system involving them is solvable. In this article we study the complexity of deciding whether a tuple of congruences of a finite algebra is a Chinese Remainder tuple. This problem, which we denote \prob{CRT}, is easily seen to lie in \textup{coNP}. We prove that it is actually \textup{coNP}-complete and also show that it is tractable when restricted to several well-known classes of algebras, such as vector spaces and distributive lattices.
The polynomial algorithms we exhibit are made possible by purely algebraic characterizations of Chinese Remainder tuples for algebras in these classes, which constitute interesting results in their own right. Among these, an elegant characterization of Chinese Remainder tuples of finite distributive lattices stands out. Finally, we address the restriction of \prob{CRT} to an arbitrary equational class $\mathcal{V}$ generated by a two-element algebra. Here we establish an (almost) dichotomy by showing that, unless $\mathcal{V}$ is the class of semilattices, the problem is either \text{coNP}-complete or tractable.
\end{abstract}

\section{Introduction}

A fundamental property of the integers is the so-called Chinese Remainder Theorem, which can be stated as follows.

\begin{CRT}
Let $m_1,\ldots,m_k$ be positive integers, and let $a_1,\ldots,a_k$ be arbitrary integers such that
\[
a_i \equiv a_j \mod{g_{ij}} \text{ for all } i,j \in \{1,\ldots,k\},
\]
where $g_{ij} := \gcd(m_i,m_j)$. Then, there is an integer $x$ such that
\begin{align*}
x & \equiv a_1 \mod{m_1}, \\
x & \equiv a_2 \mod{m_2}, \\
& \hspace{1.7mm} \vdots \\
x & \equiv a_k \mod{m_k}.
\end{align*}
\end{CRT}

This statement can be recast in the language of Universal Algebra. For this, recall that an {\em algebra} is just a structure of a first-order language with no relation symbols, and a {\em congruence} on an algebra is an equivalence relation on its base set compatible with its fundamental operations. Since the set of all congruences of an algebra $\mathbf{A}$ is closed under arbitrary intersections, it can be endowed with a complete lattice structure; the resulting lattice is usually denoted by $\Conn \mathbf{A}$. The join $\theta \vee \delta$ of two congruences in $\Conn \mathbf{A}$ is the transitive closure of $\theta \cup \delta$. In this setting, let us say that a tuple of congruences $\theta_1,\ldots,\theta_k$ of an algebra $\mathbf{A}$ is a {\em Chinese Remainder tuple} (CR tuple for short) of $\mathbf{A}$ provided that for all $a_1,\ldots,a_k$ from $\mathbf{A}$ such that $\langle a_i,a_j\rangle \in \theta_i \vee \theta_j$ for $i,j \in \{1,\ldots,k\}$ there is $a$ from $\mathbf{A}$ with $\langle a,a_i\rangle \in \theta_i$ for $i \in \{1,\ldots,k\}$. Thus, considering the integers as an algebra $\mathbf{Z}$ in the language $\{+,\cdot,-,0\}$, its congruences are exactly the equivalence relations modulo positive integers, and the original Chinese Remainder Theorem says that every tuple of congruences of $\mathbf{Z}$ is a CR tuple. It turns out that this property of the ring of integers is actually captured by two particular conditions on its congruences: they form a distributive lattice, and they permute, that is, the join of any pair of congruences is their relational composition. Any algebra whose congruences enjoy these two properties is called {\em arithmetic}, and Pixley showed that an algebra $\mathbf{A}$ is arithmetic if and only if every tuple of congruences of $\mathbf{A}$ is a CR tuple (see \cite{Pixley72-CompletenessArithmeticalAlg}). Since every ring has permutable congruences, it follows that any congruence-distributive ring is arithmetic, thus, Pixley's theorem yields the Chinese Remainder Theorem for $\mathbf{Z}$. Other examples of arithmetic algebras are Boolean algebras, or more generally any algebra in a discriminator variety \cite{BurSan81-Book-CourseUnivAlg}, and residuated lattices \cite{GalatosEtAl07-Book-RL} (e.g., Heyting algebras). On the other hand, arithmeticity is a very strong condition on an algebra, and there are many well-known structures that fail to be arithmetic and, thus, have tuples of congruences that are not CR tuples. For instance, in the three-element lattice with $0 < \frac{1}{2} < 1$ the pair of congruences given by the partitions $\{\{0\}, \{\frac{1}{2},1\}\}$ and $\{\{0,\frac{1}{2}\},\{1\}\}$ is not a CR tuple.

A natural computational problem that arises from this discussion is the following.

\begin{center}
\fbox{
\begin{minipage}{10.6cm}
\framebox{\prob{CRT}} 

\medskip

{\em Instance}: A finite algebra $\mathbf{A}$ and $\theta_1,\ldots,\theta_k$ congruences of $\mathbf{A}$.

\medskip

{\em Question}: Is $\langle \theta_1,\ldots,\theta_k\rangle$ a CR tuple of $\mathbf{A}$?
\end{minipage}
}
\end{center}

In this paper we study the computational complexity of this problem and its restrictions to several familiar classes of algebras. We prove that \prob{CRT} is \textup{coNP}-complete (see Section \ref{SEC: Hardness}). We also show that there are polynomial-time algorithms that decide \prob{CRT} for inputs in the following three classes of algebras: vector spaces (see Section \ref{SUBSEC: Vector spaces}), distributive nearlattices (see Section \ref{SUBSEC: Distributive nearlattices}) and algebras belonging to some dual discriminator variety (see Section \ref{SUBSEC: Dual discriminator varieties}). In each of the latter two cases the algorithm is powered by a purely algebraic characterization of CR tuples, which constitutes an interesting result in its own right. In particular, we obtain elegant characterizations of CR tuples in finite distributive lattices  and Tarski algebras (Corollaries \ref{CORO: caract CR tuple for DIST LAT} and \ref{CORO: caract CR tuple for TARSKI}). The article concludes (Section \ref{SEC: Almost Dichotomy}) with the study of all varieties (i.e., classes defined by identities) generated by a two-element algebra. We were able to show that the restriction of \prob{CRT} to each of these varieties is either \textup{coNP}-complete or in \textup{P} for all cases but for the variety of semilattices.

\section{Preliminaries}
\label{SEC: Preliminaries}

In this section we fix some notation and terminology. We assume the reader is familiar with the basic concepts of Universal Algebra and Complexity Theory. For undefined notions see, e.g., \cite{BurSan81-Book-CourseUnivAlg,Goldreich08-BookComplexity}.

An {\em algebra} is a model of a first-order language with no relation symbols. Algebras are usually denoted by boldface capital letters such as $\mathbf{A}, \mathbf{B}, \mathbf{C}$, etc., whereas their underlying sets are denoted by the corresponding non-bold letters, e.g., $A$ is the universe of $\mathbf{A}$.

Let $\mathbf{A}$ be an algebra. A {\em congruence} of $\mathbf{A}$ is an equivalence relation $\theta$ on $A$ compatible with the operations of $\mathbf{A}$, that is, for each positive integer $n$, each $n$-ary function symbol $f$ and elements $a_1,\ldots,a_n,b_1,\ldots,b_n \in A$ such that $\langle a_1,b_1\rangle, \ldots, \langle a_n,b_n\rangle \in \theta$ we have that $\langle f^\mathbf{A}(a_1,\ldots,a_n),f^\mathbf{A}(b_1,\ldots,b_n)\rangle \in \theta$. Since the set $\Con \mathbf{A}$ of all congruences of $\mathbf{A}$ is closed under arbitrary intersections, it has an associated complete lattice denoted by $\Conn \mathbf{A}$. The least element in $\Conn \mathbf{A}$, is the equality relation and is denoted by $\Delta_A$ (or just $\Delta$ if $A$ is clear from the context). The join $\theta \vee \delta$ of two congruences in $\Conn \mathbf{A}$ is the transitive closure of $\theta \cup \delta$.

A {\em congruence system} on $\mathbf{A}$ is a tuple $\langle \theta_1,\ldots,\theta_k,a_1,\ldots,a_k\rangle$ where $\theta_1,\ldots,\theta_k \in \Con \mathbf{A}$, $a_1,\ldots,a_k \in A$ and $\langle a_i,a_j\rangle \in \theta_i \vee \theta_j$ for all $i,j \in \{1,\ldots,k\}$. A {\em solution} to the congruence system $\langle \theta_1,\ldots,\theta_k,a_1,\ldots,a_k\rangle$ is an element $a \in A$ satisfying $\langle a,a_i\rangle \in \theta_i$ for all $i \in \{1,\ldots,k\}$. A tuple of congruences $\langle \theta_1,\ldots, \theta_k\rangle$ is said to be a {\em Chinese Remainder tuple} (CR tuple for short) of $\mathbf{A}$ provided that every system $\langle \theta_1,\ldots,\theta_k,a_1,\ldots,a_k\rangle$ has a solution.

In this article we investigate the complexity of the following computational problem.
\begin{center}
\fbox{
\begin{minipage}{10.6cm}
\framebox{\prob{CRT}} 

\medskip

{\em Instance}: A finite algebra $\mathbf{A}$ and $\theta_1,\ldots,\theta_k$ congruences of $\mathbf{A}$.

\medskip

{\em Question}: Is $\langle \theta_1,\ldots,\theta_k\rangle$ a CR tuple of $\mathbf{A}$?
\end{minipage}
}
\end{center}

Throughout this work we assume a fixed encoding of finite algebras and congruences as strings of a finite alphabet. Since any reasonable choice of such an encoding makes no difference in the arguments below, we do not take the trouble to go into the low-level details, and the exposition is entirely presented at an abstract level.

All polynomial reductions in this article are Karp reductions, that is, many-one poly-time computable reductions.

Here is an easy observation.

\begin{lemma}
\prob{CRT} is in \textup{coNP}.
\end{lemma}

\begin{proof}
Note that the fact that $\langle \theta_1, \ldots, \theta_k\rangle$ is not a CR tuple of an algebra $\mathbf{A}$ is witnessed by a tuple $\langle a_1,\ldots,a_k\rangle \in A^k$ such that $\langle \theta_1, \ldots, \theta_k, a_1, \ldots, a_k\rangle$ is a system and $a_1/\theta_1 \cap \ldots \cap a_k/\theta_k = \emptyset$. Since both these tests can be carried out in polynomial time with respect to the input, we are done.
\end{proof}

As we shall see in the next section, \prob{CRT} is actually complete for \text{coNP}. Much of this article is concerned with identifying classes of algebras for which \prob{CRT} is tractable. To this end it is helpful to introduce the following notation. Given a class of algebras $\mathcal{K}$, we write \prob{CRT$|_\mathcal{K}$} to denote the restriction of \prob{CRT} to instances where the algebra $\mathbf{A}$ is in $\mathcal{K}$. A particular type of class is prominent in the sequel. Recall that a {\em variety} is a class of algebras in the same language axiomatized by identities.

We close this section with a basic fact about CR tuples whose easy proof is left to the reader. Given congruences $\delta \subseteq \theta$ of an algebra $\mathbf{A}$, recall that $\theta/\delta := \{\langle a/\delta,b/\delta\rangle: \langle a,b\rangle \in \theta\}$ is a congruence of $\mathbf{A}/\delta$.

\begin{lemma} \label{LEMA: CRT <-> CRT diagonal}
Let $\mathbf{A}$ be an algebra, and let $\theta_1,\ldots, \theta_k$ be congruences of $\mathbf{A}$. Put $\delta := \theta_1 \cap \ldots \cap \theta_k$. The following are equivalent: 
\begin{enumerate}[\rm (a)]
\item $\langle \theta_1, \ldots, \theta_k\rangle$  is a CR tuple of $\mathbf{A}$;
\item $\langle \theta_1/\delta, \ldots, \theta_k/\delta\rangle$ is a CR tuple of $\mathbf{A}/\delta$.
\end{enumerate}
\end{lemma}

\section{\prob{CRT} is \textup{coNP}-complete}
\label{SEC: Hardness}

In this section we establish the \textup{coNP}-hardness of \prob{CRT} by providing a reduction from a special form of \prob{3SAT} to the complement of \prob{CRT}.

Given a propositional formula $\varphi$ let $\var(\varphi)$ be the set of variables occurring in $\varphi$. If $\varphi$ is in CNF, we write $\varsets(\varphi)$ for the set $\{\var(\gamma): \gamma \text{ is a clause of } \varphi\}$.

Let \prob{3SAT$'$} be the restriction of \prob{3SAT} to input formulas $\varphi$ satisfying the following: 
\begin{enumerate}[(C1)]
\item each clause of $\varphi$ has three pairwise different variables,
\item \label{C: al menos 5} $|\varsets(\varphi)| \geq 5$,
\item each variable of $\varphi$ occurs in at least three distinct sets in $\varsets(\varphi)$.
\end{enumerate}
It is an easy exercise to show that \prob{3SAT$'$} is \textup{NP}-complete.

We are now ready to present the main result of this section.

\begin{theorem} \label{TEO: hardness para sets}
\prob{CRT} is \textup{coNP}-complete.
\end{theorem}

\begin{proof}
We prove that \prob{3SAT$'$} is polynomially reducible to the complement of \prob{CRT}.

Given an instance $\varphi$ of \prob{3SAT$'$} we construct a set $S$ and equivalence relations $\theta_1, \ldots, \theta_k$ on $S$ such that $\langle \theta_1,\ldots,\theta_k\rangle$ is a CR tuple if and only if $\varphi$ is not satisfiable; more precisely, such that a system $\langle \theta_1,\ldots,\theta_k,s_1,\ldots,s_k\rangle$ is not solvable if and only if $\langle s_1,\ldots,s_k\rangle$ corresponds to an assignment making $\varphi$ true.

Put $\{V_1,\ldots,V_k\} := \varsets(\varphi)$; note that $k \geq 5$ by (C\ref{C: al menos 5}). For $i \in \{1,\ldots,k\}$ define
\[
A_i := \{a \in \{0,1\}^{V_i}: \gamma(a) = 1 \text{ for all clauses $\gamma$ in $\varphi$ with $\var(\varphi) = V_i$}\},
\]
that is, $A_i$ is the set of assignments to the variables in $V_i$ that make all clauses of $\varphi$ whose variables are in $V_i$ true.   Next, put
\[
A := \bigcup_{i=1}^k A_i.
\] Given $a \in A$ we write $\dom(a)$ for the set of variables assigned by $a$; note that $\dom(a) \in \varsets(\varphi)$. For a set $V$ of variables and $a,b \in A$ we say that $a$ and $b$ are {\em $V$-compatible}, in symbols $a \uparrow_V b$, if they assign the same values to the variables in $V \cap \dom(a) \cap \dom(b)$. We simply write $a \uparrow b$ when $a$ and $b$ agree on all variables in $\dom(a) \cap \dom(b)$; in this case, we say that $a$ and $b$ are {\em compatible}. We are ready to define our base set; let
\[
S := \{\{a,b\} \subseteq A: a \uparrow b\}.
\]
Observe that $\{a\} \in S$ for every $a \in A$. In addition, if $\{a,b\} \in S$ and $a \ne b$, then $\dom(a) \ne \dom(b)$.

The final phase in the construction is to define the congruences. For each $i \in \{1,\ldots,k\}$ let
\begin{equation} \label{EQ: def thetas}
\theta_i := \{\langle s,t\rangle \in S^2: s \cap t = \{a\} \text{ for some }a \in A_i\} \cup \{\langle s,s\rangle: s \in S\}.
\end{equation}
We leave it to the reader to check that $\theta_i$ is in fact an equivalence relation. This finishes the definition of our reduction. Again we leave to the reader the straightforward verification that the reduction is many-one poly-time. 

We prove next that the reduction works. The argument is split into several claims, and the theorem follows from Claims \ref{CL: sat -> hay no soluble} and \ref{CL: no sol -> sat}.

We say that a system $\langle \theta_1,\ldots,\theta_k,s_1,\ldots,s_k\rangle$ is {\em coherent} provided that $s_i \cap A_i \neq \emptyset$ for every $i \in \{1,\ldots,k\}$; a system is {\em incoherent} if it is not coherent. 

\begin{claim} \label{CL: coherent -> no soluble}
Coherent systems are not solvable. 
\end{claim}

For the sake of contradiction suppose $\langle \theta_1,\ldots,\theta_k,s_1,\ldots,s_k\rangle$ is a coherent system with a solution $s$.  For $i \in \{1,\ldots,k\}$ take $a_i$ as the (unique) element in $s_i \cap A_i$, and observe that, since $s \mathrel{\theta_i} s_i$, it follows by \eqref{EQ: def thetas} that $a_i \in s$. But, as the $a_i$'s are all different, we have that $|s| \geq k \geq 5$, which is clearly not possible. This finishes the proof of the claim.

\begin{claim} \label{CL: sat -> hay no soluble}
If $\varphi$ is satisfiable, then there is an unsolvable system.
\end{claim}

By Claim \ref{CL: coherent -> no soluble} it suffices to check that, if $a$ is an assignment that makes $\varphi$ true, then
\begin{equation} \label{EQ: sistema}
\langle \theta_1,\ldots,\theta_k,\{a|_{V_1}\},\ldots,\{a|_{V_k}\}\rangle,
\end{equation}
where $a|_{V_i}$ stands for the restriction of $a$ to $V_i$, is a a coherent system. Observe that for $i,j \in \{1,\ldots,k\}$ we have $a|_{V_i} \uparrow a|_{V_j}$, and thus
\[
\{a|_{V_i}\} \mathrel{\theta_i} \{a|_{V_i},a|_{V_j}\} \mathrel{\theta_j} \{a|_{V_j}\}.
\]
Hence \eqref{EQ: sistema} is a system which is clearly coherent.

For the next leg of the proof it is convenient to extend the compatibility relation to $S$. Given $s,t \in S$ and a set of variables $V$ we write $s \uparrow_{V} t$ if $a \uparrow_V b$ for all $a \in s$ and  $b \in t$.

Let us record two properties of the joins of the $\theta_j$'s needed in the sequel. Both facts follows routinely from the definitions, and we leave their proofs to the reader.

\begin{claim} \label{CL: props thetas}
Let $s,t \in S$ and $i,j \in \{1,\ldots,k\}$; the following holds:
\begin{enumerate}
\item[\textup{(a)}] \label{CL: props thetas - inciso singletos supremo} If $s \cap A_i = \emptyset$, $s \ne t$ and $s \mathrel{(\theta_i \vee \theta_j)} t$, then $s \cap A_j \ne \emptyset$.

\item[\textup{(b)}] If $s \mathrel{(\theta_i \vee \theta_j)} t$, then $s \uparrow_{V_i \cap V_j} t$.
\end{enumerate}
\end{claim}

Next, we establish the key property of incoherent systems.

\begin{claim} \label{CL: incoh syst are solvable}
Incoherent systems are solvable.
\end{claim}

Suppose $\langle \theta_1,\ldots,\theta_k,s_1,\ldots,s_k\rangle$ is an incoherent system, that is, there is $i \in \{1,\ldots,k\}$ such that $s_i \cap A_i = \emptyset$. We may assume without loss that $s_1 \cap A_1 = \emptyset$. Note that Claim \ref{CL: props thetas}.(a) implies that $s_1 \cap A_\ell \ne \emptyset$ for any $\ell \in \{1,\ldots,k\}$ such that $s_\ell \ne s_1$. Thus,
\begin{equation} \label{EQ: a lo sumo dos}
\text{the set $I := \{\ell \in \{1,\ldots,k\}: s_\ell \ne s_1\}$ has at most two elements. }
\end{equation}

We prove that $s_1$ is a solution to the system. Clearly it suffices to show that $s_1 \mathrel{\theta_i} s_i$ for $i \in I$. For clarity of exposition we provide the proof for the concrete case $I = \{2,3\}$. However, it should be clear that the arguments below apply seamlessly to the general case. 

We start by observing that there are $a_2 \in A_2$ and $a_3 \in A_3$ such that
\begin{equation} \label{EQ: s1 = a2,a3}
s_1 = \{a_2,a_3\}.
\end{equation}
Also, in view of \eqref{EQ: a lo sumo dos}, the system is as follows:
\begin{equation} \label{EQ: sistema incoherente}
\langle \theta_1,\theta_2,\theta_3,\theta_4,\ldots,\theta_k,s_1,s_2,s_3,s_1,\ldots,s_1\rangle.
\end{equation}
So, to prove that $s_1$ is a solution, it suffices to establish that $s_1 \mathrel{\theta_2} s_2$ and $s_1 \mathrel{\theta_3} s_3$.  We first show that $s_2 \cap A_2 \ne \emptyset$. Aiming for a contradiction suppose this does not hold. Then, invoking Claim \ref{CL: props thetas}.(a), we have $s_2 \cap A_\ell \ne \emptyset$ for any $\ell$ such that $s_2 \ne s_\ell$. Looking at \eqref{EQ: sistema incoherente} we see that there are at least $k-2$ indices $\ell$ such that $s_2 \ne s_\ell$. But this is not possible because it would entail that $|s_2| \geq 3$ since $k \geq 5$. Hence, 
\begin{equation} \label{EQ: b2 in s2}
\text{there is $b_2 \in s_2 \cap A_2$.}
\end{equation}
We claim that $b_2 = a_2$. In fact, suppose $p$ is a variable in $V_2$. By the assumption (C3), there is $V_j \in \varsets(\varphi)$ with $j \ne 2,3$ such that $p \in V_j$. Now, the system condition says that $s_1 = s_j \mathrel{(\theta_2 \vee \theta_j)} s_2$. Now, Claim \ref{CL: props thetas}.(b) yields $s_1 \uparrow_{V_2 \cap V_j} s_2$, and as $p \in V_2 \cap V_j$ we have $a_2(p) = b_2(p)$. Since $p$ is an arbitrary variable in $V_2$, we conclude that $a_2 = b_2$. This fact in combination with \eqref{EQ: s1 = a2,a3} and \eqref{EQ: b2 in s2} produces $s_1 \cap s_2 = \{a_2\}$, which in turn says that $s_1 \mathrel{\theta_2} s_2$.

Obviously, the same line of reasoning shows that $s_1 \mathrel{\theta_3} s_3$, and hence we conclude that $s_1$ is in fact a solution to the system.

The remaining task to complete the proof of the theorem is to establish the following.

\begin{claim} \label{CL: no sol -> sat}
If there is an unsolvable system, then $\varphi$ is satisfiable.
\end{claim}

Suppose $\langle \theta_1,\ldots,\theta_k,s_1,\ldots,s_k\rangle$ has no solution. By Claim \ref{CL: incoh syst are solvable}, the system must be coherent, and thus there is $a_\ell \in s_\ell \cap A_\ell$ for each $\ell \in \{1,\ldots,k\}$. Note that Claim \ref{CL: props thetas}.(b) implies $a_i \uparrow_{V_i \cap V_j} a_j$ for $i,j \in \{1,\ldots,k\}$, which says that the assignment $a\colon \var(\varphi) \to \{0,1\}$ given by
\[
a(p) := \begin{cases} a_1(p) & \text{if } p \in V_1, \\ \quad \vdots \\ a_k(p) & \text{if }p \in V_k, \end{cases}
\]
is well defined. Finally, since for each $\ell \in \{1,\ldots,k\}$ we have $a_\ell \in A_\ell$, the partial assignment $a_\ell$ makes each clause of $\varphi$ with variables in $V_\ell$ true. Hence, the assignment $a$ makes $\varphi$ true.
\end{proof}

Recall that a {\em semigroup} is an algebra in the language $\{\cdot\}$ satisfying the associativity identity $x \cdot (y \cdot z) \approx (x \cdot y) \cdot z$. Let $\mathcal{SG}$ denote the class of semigroups. An immediate consequence of the previous theorem is the following.

\begin{corollary}
\prob{CRT$|_\mathcal{SG}$} is \textup{coNP}-complete.
\end{corollary} 

\begin{proof}
Observe that any set can be turned into a semigroup by adding the trivial product $x \cdot y := x$. This product is preserved by all equivalence relations.
\end{proof}

\section{Some tractable cases}
\label{SEC: Some tractable cases}

In this section we show that the restriction of \prob{CRT} to each of the following classes is in \textup{P}:
\begin{itemize}
\item vector spaces,
\item distributive nearlattices,
\item the class of algebras belonging to some dual discriminator variety.
\end{itemize}
It is worth mentioning that the provided polynomial algorithms are strongly dependent on the algebraic features of each of these classes, and thus each of them turns out to have a unique character.

\subsection{Vector spaces}
\label{SUBSEC: Vector spaces}

We present next a poly-time algorithm solving the restriction of \prob{CRT} to vector spaces over finite fields. Our algorithm rides piggy-back on Gaussian elimination and takes advantage of the fact that vector spaces have a dimension.

Given a finite field $\mathbf{F}$, we consider $\mathbf{F}$-vector spaces as algebras in the language $\{+, -, 0\} \cup \{\lambda_r: r \in F\}$, where each $\lambda_r$ is a unary symbol for scalar multiplication by $r$. Let $\mathcal{VS}$ denote the class of all finite vector spaces over finite fields.

\begin{theorem} \label{TEO: vector spaces is in P}
\prob{CRT$|_\mathcal{VS}$} is in \textup{P}.
\end{theorem}

\begin{proof}
Given $\mathbf{V} \in \mathcal{VS}$ and $\theta_1,\ldots,\theta_k$ congruences of $\mathbf{V}$, we can explicitly compute in polynomial time a finite field $\mathbf{F}$, a positive integer $n$, an isomorphim $\gamma\colon \mathbf{V} \to \mathbf{F}^n$ and subsets $W_1,\ldots,W_k \subseteq F^n$ such that for $i \in \{1,\ldots,k\}$ we have that $\langle v_1,v_2\rangle \in \theta_i$ iff $\gamma(v_1 - v_2) \in W_i$. Thus, $\langle \theta_1,\dots,\theta_k\rangle$ is a CR tuple of $\mathbf{V}$ if and only if 
\begin{itemize}
\item[$(*)$] for all $v_1,\dots,v_k \in F^n$ such that $v_i - v_j \in W_i + W_j := \{w + w': w \in W_i, w' \in W_j\}$ for all $i,j \in \{1,\ldots,k\}$, there is $z \in F^n$ such that $z-v_i \in W_i$ for $i \in \{1,\dots,k\}$.
\end{itemize}
We prove the theorem by showing that $(*)$ is decidable in polynomial time from the input $\langle \mathbf{F}^n, W_1,\ldots,W_k\rangle$. 

Define

$S := \{\langle z+w_1,\ldots,z+w_k\rangle: z \in F^n, w_1 \in W_1, \ldots, w_k \in W_k\}$,

$T := \{\langle v_1,\ldots,v_k\rangle \in (F^n)^k: v_i - v_j \in W_i + W_j \text{ for } i,j \in \{1,\ldots,k\}\}$.
Note that $S \subseteq T$ and both are subuniverses of $(\mathbf{F}^n)^k$. Further note that $S = T$ if and only if $(*)$ holds. Thus, since $\mathbf{S}$ and $\mathbf{T}$ are vector spaces over $\mathbf{F}$, it suffices to check whether they have the same dimension.

Since $\dim \mathbf{W} = |W|/|F|$ for any finite $\mathbf{F}$-vector space $\mathbf{W}$, the dimension of $\mathbf{S}$ is poly-time computable by observing that
\[
\dim \mathbf{S} = n + \dim \mathbf{W}_1 + \ldots + \dim \mathbf{W}_k - \dim (\mathbf{W}_1 \cap \ldots \cap \mathbf{W}_k).
\]

Next, we show that $\dim \mathbf{T}$ is also poly-time computable. It is a routine exercise to see that for each pair $i,j \in \{1,\ldots,k\}$, $i < j$,  we can construct in polynomial time an $(n-\dim(\mathbf{W}_i+\mathbf{W}_j)) \times n$ matrix $A_{ij}$ with entries in $F$ such that $A_{ij}v = 0$ iff $v \in W_i + W_j$. 
Next, using the $A_{ij}$'s as blocks, we can construct a matrix $A$ such that for $\langle v_1,\ldots,v_k\rangle \in (F^n)^k$ we have
\[
\begin{bmatrix}
 • & • & • & • & • \\ 
 • & • & • & • & • \\ 
 • & • & A & • & • \\ 
 • & • & • & • & • \\ 
 • & • & • & • & •
 \end{bmatrix}  \begin{bmatrix}
v_1 \\ 
v_2 \\
\vdots \\
v_k
\end{bmatrix} = \begin{bmatrix}
A_{12}(v_1 - v_2) \\ 
A_{13}(v_1 - v_3) \\ 
\vdots \\ 
A_{1k}(v_1 - v_k) \\ 
A_{23}(v_2 - v_3) \\ 
\vdots \\ 
A_{2k}(v_2 - v_k) \\ 
\vdots \\ 
A_{(k-1)k}(v_{k-1}-v_k)
\end{bmatrix}.
\]
Since each of the $A_{ij}$'s is poly-time computable, it is clear we can construct $A$ in polynomial time. Finally, observe that $T$ agrees with the kernel of $A$, and thus we can compute the dimension of $\mathbf{T}$ in polynomial time by applying Gaussian elimination to $A$.
\end{proof}

\subsection{Distributive nearlattices}
\label{SUBSEC: Distributive nearlattices}

We deal now with distributive nearlattices, which constitute a common generalization of distributive lattices and Tarski algebras (see \cite{Abbott67-ImplicationAlgebras} for a definition). Since distributive nearlattices are unlikely to be familiar to the reader, we recall next their definition and some basic facts.

Let $\mathbf{2}_\mathcal{N} := \langle \{0,1\},\mathsf{n}\rangle$ where $\mathsf{n}$ is the ternary operation defined by $\mathsf{n}(x,y,z) := (x \wedge y) \vee z$ (of course, $\wedge$ and $\vee$ are the natural meet and join operation on the set $\{0,1\}$). Let $\mathcal{N}$ be the variety generated by $\mathbf{2}_\mathcal{N}$.  The algebras in $\mathcal{N}$ are called {\em distributive nearlattices}. 
 It is known that this variety is congruence-distributive and $\mathbf{2}_\mathcal{N}$ is, up to isomorphisms, the only subdirectly irreducible member of $\mathcal{N}$ (\cite{Halas06-SInearlattices}).
 
Let $\mathbf{A} \in \mathcal{N}$. Note that $\mathbf{A}$ is a join-semilattice with the operation $x \vee y := \mathsf{n}\langle x,x,y\rangle$ satisfying that for all $a \in A$ the up-set $[a)$ is a distributive lattice under the order induced by $\vee$. In the sequel we freely refer to this order and we denote it by $\leq$. Also, given $a,b \in A$ with a common lower bound we have that the meet $a \wedge b$ exists.  An element $p \in A$ is {\em meet-irreducible} provided that $p$ is not the top element of $\mathbf{A}$ and whenever $p = a \wedge b$ with $a,b \in A$,  we have $p = a$ or $p = b$. Replacing $=$ by $\geq$ in the previous sentence we get the definition of a {\em meet-prime} element. As for distributive lattices, an element is meet-irreducible if and only if it is meet-prime. 

Until further notice we assume $\mathbf{A}$ is finite. Denote by $P_\mathbf{A}$ the set of meet-irreducible elements of $\mathbf{A}$, and let $\mathbf{P}_\mathbf{A}$ be the poset obtained by ordering $P_\mathbf{A}$ by $\leq$. Let $D_\mathbf{A}$ be the set of down-sets of $\mathbf{P}_\mathbf{A}$. Of course, $D_\mathbf{A}$ has a natural distributive lattice structure under the inclusion ordering. Hence, the operation $\mathsf{n}(x,y,z) := (x \cap y) \cup z$ makes it into a distributive nearlattice, which we denote by $\mathbf{D}_\mathbf{A}$.

Given $a \in A$ define
\[
\sigma(a) := \{p \in P_\mathbf{A}: p \not\geq a\}.
\]
It is not hard to show that the map $a \mapsto \sigma(a)$ is a nearlattice embedding from $\mathbf{A}$ into $\mathbf{D}_\mathbf{A}$. We write $A^\sigma$ for the range of $\sigma$ and $\mathbf{A}^\sigma$ for the corresponding substructure of $\mathbf{D}_\mathbf{A}$. We shall need the following two facts, whose proofs we leave to the reader:
\begin{itemize}
\item $\mathbf{A}^\sigma$ is an up-set of $\mathbf{D}_\mathbf{A}$;
\item $\bigcap A^\sigma = \emptyset$.
\end{itemize}

We next address the congruences of $\mathbf{A}$. Given $p \in P_\mathbf{A}$ let 
\[
\theta_p := \{\langle a,b\rangle \in A^2: a \leq p \Leftrightarrow b \leq p\}.
\]
A routine argument shows that $\theta_p$ is a meet-irreducible congruence of $\mathbf{A}$, and that, furthermore, every meet-irreducible congruence of $\mathbf{A}$ is of this form. Now, since $\mathbf{A}$ is semisimple and $\Con \mathbf{A}$ is distributive, we have a correspondence between congruences of $\mathbf{A}$ and sets of meet-irreducible elements of $\mathbf{A}$. Thus, the map
\[
\theta \mapsto F_\theta := \{p \in P_\mathbf{A}: \theta \subseteq \theta_p\}.
\]
is an order-reversing bijection from the congruence lattice of $\mathbf{A}$
onto the powerset of $P$ ordered by inclusion. The inverse of this map is given by
\[
F \mapsto \theta_F := \{\langle a,b\rangle\in A^{2}: \sigma(a) \cap F = \sigma(b) \cap F\}.
\]
In particular we have $\theta_{F}\vee\theta_{F'}=\theta_{F\cap F'}$
and $\theta_{F}\cap\theta_{F'}=\theta_{F\cup F'}$, which entails
the following easy observation: $\langle \theta_{1},\ldots,\theta_{k},a_{1},\ldots,a_{k}\rangle$
is a system on $\mathbf{A}$ if and only if 
\begin{equation}
\sigma(a_i) \cap F_{i} \cap F_{j} = \sigma(a_{j}) \cap F_{i} \cap F_{j}\text{ for } i,j \in \{1,\ldots,k\},\label{eq:sistema de Fs}
\end{equation}
where $F_{i}:=F_{\theta_{i}}$. In addition, an element $a$ is a
solution for the system iff 
\begin{equation}
\sigma(a) \cap F_{i} = \sigma(a_{i}) \cap F_{i} \text{ for } i \in \{1,\ldots,k\}.\label{eq:sol a siste de Fs}
\end{equation}

The next lemma characterizes the solvable systems.

\begin{lemma}
\label{LEMA: s en L <-> sol} Let $\mathbf{A}$ be a finite algebra in $\mathcal{N}$, and let $\langle \theta_1, \ldots, \theta_k,a_1,\ldots,a_k\rangle$ be a system on $\mathbf{A}$ such that $\bigcap_{i=1}^{k}\theta_{i} = \Delta$. Then, the set
\[
s:=\bigcup_{i=1}^{k} \sigma(a_{i}) \cap F_{\theta_i}
\]
is the unique subset $x$ of $P_\mathbf{A}$ satisfying $x\cap F_{\theta_i} =\sigma(a_{i}) \cap F_{\theta_i}$
for all $i\in\{1,\ldots,k\}$. Thus, the system is solvable if and
only if $s\in A^\sigma$.
\end{lemma}

\begin{proof}
Put $F_i := F_{\theta_i}$ for $i \in \{1,\ldots,k\}$. Observe that $\bigcup_{i=1}^k F_i = P_\mathbf{A}$ since $\bigcap_{i=1}^k \theta_i = \Delta$. We check first that $s$ satisfies $s\cap F_{i}=\sigma(a_{i}) \cap F_{i}$
for all $i\in\{1,\ldots,k\}$. Fix $j\in\{1,\ldots,k\}$ and note
that:
\begin{align*}
s\cap F_{j} & =\bigcup_{i=1}^{k}\sigma(a_{i})\cap F_{i}\cap F_{j}\\
 & =\bigcup_{i=1}^{k}\sigma(a_{j})\cap F_{i}\cap F_{j}\\
 & =\sigma(a_{j})\cap F_{j}.
\end{align*}

Next, suppose $x\subseteq P_\mathbf{A}$ satisfies $x\cap F_{i}=\sigma(a_{i})\cap F_{i}$
for all $i\in\{1,\ldots,k\}$. Then, $\bigcup_{i=1}^{k}x\cap F_{i}=\bigcup_{i=1}^{k}\sigma(a_{i}) \cap F_{i}$,
which produces $x=s$.
\end{proof}

Next we single out some special elements in $D_\mathbf{A}$ (recall that $D_\mathbf{A}$
stands for the set of down-sets of $\mathbf{P}_\mathbf{A}$). Let $b \in D_\mathbf{A}$ and
$F \subseteq P_\mathbf{A}$, we say that:
\begin{itemize}
\item $b$ is a \emph{fringe} element of $\mathbf{A}$ if $b$ is maximal (with respect to inclusion) in $D_\mathbf{A} \setminus A^\sigma$;
\item $b$ is $F$\emph{-interpolable} by $\mathbf{A}$ if there is $a \in A$ such that $b\cap F = \sigma(a) \cap F$.
\end{itemize}
It follows from our definitions and the fact that $A^\sigma$ is an up-set of $\mathbf{D}_\mathbf{A}$ that if $b$ is $F$-interpolable
by $\mathbf{A}$ and $b'\in D_\mathbf{A}$ is such that $b\subseteq b'$, then $b'$ is $F$-interpolable by $\mathbf{A}$.

We are now in a position to characterize CR tuples.
\begin{theorem}
\label{TEO: CRT-tuple for dist near  lat - equivalence} Let $\mathbf{A}$ be a finite algebra in $\mathcal{N}$, and let $\theta_{1},\ldots,\theta_{k}$ be congruences on $\mathbf{A}$ such that $\bigcap_{i=1}^{k} \theta_{i} = \Delta$. The following are equivalent:
\begin{enumerate}
\item $\langle \theta_{1},\ldots,\theta_{k}\rangle$ is a CR tuple of $\mathbf{A}$.
\item The following holds:
\begin{enumerate}
\item for all $p,q\in P_\mathbf{A}$, such that $p$ covers $q$, there is $\ell\in\{1,\ldots,k\}$ with $p,q\in F_{\theta_\ell}$;
\item for every fringe element $b$ of $\mathbf{A}$ there is $\ell\in\{1,\ldots,k\}$ such that $b$ is not $F_{\theta_\ell}$-interpolable by $\mathbf{A}$.
\end{enumerate}
\end{enumerate}
\end{theorem}

\begin{proof}
As in the proof of Lemma \ref{LEMA: s en L <-> sol} put $F_i := F_{\theta_i}$ for $i \in \{1,\ldots,k\}$ and note that $\bigcup_{i=1}^k F_i = P_\mathbf{A}$.

2 $\Rightarrow$ 1. Suppose $\langle \theta_{1},\ldots,\theta_{k},a_{1},\ldots,a_{k}\rangle$
is a system. By Lemma \ref{LEMA: s en L <-> sol}, we have to show that
$s:=\bigcup_{i=1}^{k} \sigma(a_i) \cap F_{i}$ is in $A^\sigma$.
We establish first that $s$ is an down-set of $\mathbf{P}_\mathbf{A}$. Fix $p\in s$ and $q\in P_\mathbf{A}$ with $q \leq p$. Let $j\in\{1,\ldots,k\}$ such that
$p\in \sigma(a_j) \cap F_{j}$. We first consider the case in which $p$
covers $q$. By (2).(a), there is $\ell\in\{1,\ldots,k\}$ such that
$p,q\in F_{\ell}$. Since $p \in \sigma(a_j) \cap F_{\ell} \cap F_{j} = \sigma(a_\ell) \cap F_{\ell}\cap F_{j}$,
we have $p\in \sigma(a_{\ell})$. So, as $\sigma(a_\ell)$ is a down-set, we have $q \in \sigma(a_\ell)$, and hence $q\in \sigma(a_\ell) \cap F_{\ell} \subseteq s$.
Now an obvious inductive argument takes care of the general case.

By Lemma \ref{LEMA: s en L <-> sol} we know that $s$ is $F_{i}$-interpolable by $\mathbf{A}$ for $i\in\{1,\ldots,k\}$, and thus every element of $D_\mathbf{A}$ that contains $s$ has the same property. So, it must be that $s\in A^\sigma$, since otherwise there would be a fringe element $b \supseteq s$ refuting
(2).(b).

1 $\Rightarrow$ 2. We show that if either (2).(a) or (2).(b) fails,
then there is a system with no solution. Suppose (2).(a) does not
hold, that is, there are $p,q \in P_\mathbf{A}$ such that $q$ covers $p$, but there is no $\ell\in\{1,\ldots,k\}$
with $p,q\in F_{\ell}$. Define $I:=\{i\in\{1,\ldots,k\}:p\in F_{i}\}$
and $I':=\{1,\ldots,k\}\setminus I$; note that both $I$ and $I'$
are nonempty. Next define:
\begin{itemize}
\item $E:=\bigcup\{F_{i}:i\in I\}$ and $E':=\bigcup\{F_{i}:i\in I'\}$;
\item $a := \bigwedge ([p) \cap E\cap E')$.
\end{itemize}
Notice that, since the meet defining $a$ is computed in $[p)$, it is in fact an element of $A$. Also, note that $\sigma(a) = P_\mathbf{A} \setminus [[p) \cap E \cap E')$.  Now, let $a_{i}:= p$ for $i\in I$ and $a_{i}:= a$ for $i\in I'$. 

We prove via \eqref{eq:sistema de Fs} that $\langle \theta_{1},\ldots,\theta_{k},a_{1},\ldots,a_{k}\rangle$
is a system; this amounts to show that given $i \in I$ and $j \in I'$ we have $\sigma(p) \cap F_i \cap F_j = \sigma(a) \cap F_i \cap F_j$. The inclusion from left to right is obvious since $p \leq a$. To see the remaining inclusion let $r \in \sigma(a) \cap F_i \cap F_j$, so $r \notin [p) \cap E \cap E'$. Since $r \in F_i \cap F_j \subseteq E \cap E'$, it must be that $r \notin [p)$, that is, $r \in \sigma(p)$. 

We claim that the system $\langle \theta_{1},\ldots,\theta_{k},a_{1},\ldots,a_{k}\rangle$ has no solution. In fact, we show that
$q \in s := \bigcup_{i=1}^{k}\sigma(a_i) \cap F_{i}$ but $p \notin s$, which entails that $s$ is not a down-set, and thus the system is not solvable by Lemma
\ref{LEMA: s en L <-> sol}. First, since $p \notin \sigma(p)$ and $p \notin F_i$ for $i \in I'$, we have that $p \notin s$. We show that $q \in s$ by way of contradiction. Suppose $q \notin s$.
By our initial assumption, we know there is $j \in I'$ such that $q \in F_j$. Since $q \notin \sigma(a) \cap F_j$, we have that $q \notin \sigma(a)$, hence $q \in [[p) \cap E \cap E')$. Let $r \in E\cap E'$ be such that $p\leq r\leq q$. This
is not possible, since $p,q\notin E\cap E'$ and $q$ covers $p$.

Suppose next that (2).(b) does not hold. Then there are a fringe element $b$ and $a_{1},\ldots,a_{k}\in A$ such that $b \cap F_{i}= \sigma(a_{i}) \cap F_{i}$ for $i\in\{1,\ldots,k\}$. Clearly $\langle \theta_{1},\ldots,\theta_{k},a_{1},\ldots,a_{k}\rangle$
is a system, and $\bigcup_{i=1}^{k} \sigma(a_i) \cap F_{i} = b \notin A^\sigma$.
So, by Lemma \ref{LEMA: s en L <-> sol}, the system has no solution.
\end{proof}

\begin{theorem} \label{TEO: CRT.N is in P}
\prob{CRT$|_\mathcal{N}$} is in \textup{P}.
\end{theorem}

\begin{proof}
Suppose $\mathbf{A}, \theta_1, \ldots, \theta_k$ is an instance of \prob{CRT$|_\mathcal{N}$}. By Lemma \ref{LEMA: CRT <-> CRT diagonal} we may assume without loss that $\bigcap_{i=1}^k \theta_i = \Delta$, and hence apply Theorem \ref{TEO: CRT-tuple for dist near  lat - equivalence}. It should be clear that checking condition 2 of that theorem can be done in polynomial time once the set $B$ of fringe elements is known. Given any $s \subseteq P_\mathbf{A}$ we can decide in poly-time whether $s \in B$ by checking if $s$ is a down-set and for every $p \in P_\mathbf{A}$ minimal in $P_\mathbf{A} \setminus s$ the down-set $s \cup \{p\}$ belongs to $A^\sigma$. Finally, since $B \subseteq \tilde{B} := \{\sigma(a) \setminus \{p\}: a \in A, p \in P_\mathbf{A}\}$, we can compute $B$ by performing the above described test to every element in $\tilde{B}$.
\end{proof}

Theorem \ref{TEO: CRT-tuple for dist near  lat - equivalence} easily produces two remarkable results for distributive lattices and Tarski algebras. We omit the details, but let us notice that if $\mathbf{L}$ is a distributive lattice, then $\sigma$ is onto $D_\mathbf{L}$ and thus, condition 2.(b) in Theorem \ref{TEO: CRT-tuple for dist near  lat - equivalence} trivially holds. Similarly, if $\mathbf{A}$ is a Tarski algebra, then $\mathbf{P}_\mathbf{A}$ is totally order-disconnected and hence condition 2.(a) in Theorem \ref{TEO: CRT-tuple for dist near  lat - equivalence} is vacuous. 

\begin{corollary} \label{CORO: caract CR tuple for DIST LAT}
Let $\mathbf{L}$ be a finite distributive lattice, $\theta_{1},\ldots,\theta_{k} \in \Con \mathbf{L}$, and let $Q := \bigcup_{i=1}^k F_{\theta_i}$. The following are equivalent:
\begin{enumerate}
\item $\langle \theta_{1},\ldots,\theta_{k}\rangle$ is a CR tuple of $\mathbf{A}$.
\item For all $p,q \in Q$, such that $p$ covers $q$ in $\langle Q,\leq\rangle$, there is $\ell\in\{1,\ldots,k\}$ with $p,q\in F_{\theta_\ell}$.
\end{enumerate}
\end{corollary}

Let $\mathbf{A}$ be a finite Tarski algebra, and let $Q \subseteq P_\mathbf{A}$. An element $b \in D_\mathbf{A}$ is a {\em $Q$-fringe} element of $\mathbf{A}$ if $b$ is maximal with the property of not belonging to $A^\sigma$ and containing $P_\mathbf{A} \setminus Q$.

\begin{corollary} \label{CORO: caract CR tuple for TARSKI}
Let $\mathbf{A}$ be a finite Tarski algebra, $\theta_{1},\ldots,\theta_{k} \in \Con \mathbf{A}$, and let $Q := \bigcup_{i=1}^k F_{\theta_i}$.  The following are equivalent:
\begin{enumerate}
\item $\langle \theta_{1},\ldots,\theta_{k}\rangle$ is a CR tuple of $\mathbf{A}$.
\item For every $Q$-fringe element $b$ of $\mathbf{A}$ there is $\ell\in\{1,\ldots,k\}$ such that $b$ is not $F_{\theta_\ell}$-interpolable by $\mathbf{A}$.
\end{enumerate}
\end{corollary}

\subsection{Dual discriminator varieties}
\label{SUBSEC: Dual discriminator varieties}

The class we consider next, albeit natural, takes some work to define. It is the aggregate of all dual discriminator varieties. 

Recall that the {\em dual discriminator} on a set $A$ is the function $q_A\colon A^3 \to A$ given by $$q_A(a,b,c) := \begin{cases} a & \text{ if } a = b, \\ c & \text{ if } a \ne b. \end{cases}$$

A ternary term is a {\em dual discriminator term} for an algebra $\mathbf{A}$ provided that it interprets as the dual discriminator function on $\mathbf{A}$. A variety $\mathcal V$ is a {\em dual discriminator variety} if it is generated as a variety by a subclass $\mathcal K\subseteq \mathcal V$ with a common dual discriminator term. If $\mathcal{V}$ is a  dual discriminator variety, then $\mathcal{V}$ is semisimple and its class of simple members is the largest subclass of $\mathcal{V}$ possessing a common dual discriminator term (see \cite{FriedPixley-DualDiscrimnator}).

A standard example of a dual discriminator
variety is the variety of distributive lattices. In this case we can take $\mathcal{K}$ to be the class of two-element lattices, and the term
$(x \wedge y) \vee (x \wedge z) \vee (y \wedge z)$ interprets as the dual discriminator on $\mathcal K$.

Let $\mathcal{D}$ be the class of finite algebras belonging to some dual discriminator variety. In this section we prove that \prob{CRT$|_\mathcal{D}$} is in \textup{P} (see Theorem \ref{TEO: CRT.D is in P}). This is accomplished through a deep analysis of congruence systems for algebras in $\mathcal{D}$, which produces a nontrivial characterization of CR tuples for algebras in $\mathcal{D}$ (see Theorem \ref{TEO: CRT tuple for dual disc - equivalences}).

\subsubsection{Crosses in products}

One of the key features of dual discriminator varieties is that the subalgebras of twofold products of simple algebras come in three flavors: graphs of homomorphisms, subproducts and crosses (see Lemma \ref{LEMA: subalg cuadrado - disc dual}). Our study of congruence systems in $\mathcal{D}$ is carried out by a purely combinatorial analysis of the interaction of crosses. This analysis is presented in this section and applied to congruence systems in the next section.

Let $A_1, \ldots, A_n$ be sets. A subset $C \subseteq A_1 \times \ldots \times A_n$ is a {\em cross} ({\em of} $A_1 \times \ldots \times A_n$) provided that there are $i,j \in \{1,\ldots,n\}$ with $i \ne j$, $a \in A_i$ and $b \in A_j$ such that
\[
C = \{\bar{a} \in A_1 \times \ldots \times A_n: a_i = a \text{ or } a_j = b\}.
\]
We denote this cross by $\cross{i}{a}{j}{b}$; $i$ and $j$ are the {\em indices} of $C$.

Two crosses $C$ and $D$ are {\em compatible} if they share exactly one index, say $k$, and the elements corresponding to $k$ in $C$ and $D$ are different. For example, the crosses $\cross{i}{a}{k}{c}$ and $\cross{k}{c'}{j}{b}$ are compatible if and only if $c \ne c'$ and $i \ne j$.

We define the operation of {\em composition} between compatible crosses by
\[
\cross{i}{a}{k}{c} \circ \cross{k}{c'}{j}{b} :=
\cross{i}{a}{j}{b}.
\]
Given a set of crosses $\mathcal{C}$, let $\mathcal{C}^\circ$ be
\[
\mathcal{C} \cup \{C \circ D: \text{ compatible } C,D \in \mathcal{C}\},
\]
and define
\[
\bar{\mathcal{C}} := \mathcal{C} \cup \mathcal{C}^\circ \cup \mathcal{C}^{\circ\circ} \cup \ldots
\]

Until further notice we fix sets $A_1, \ldots, A_n$ with at least two elements and put $P := A_1 \times \ldots \times A_n$.

\begin{lemma} \label{LEMA: cruces - direccion facil}
Given a set of crosses $\mathcal{C}$ of $P$ and $C \in \bar{\mathcal{C}}$ we have that $\bigcap \mathcal{C} \subseteq C$. In particular, $\bigcap \bar{\mathcal{C}} = \bigcap \mathcal{C}$.
\end{lemma}

\begin{proof}
This is an easy induction.
\end{proof}

If $S$ is a subset of $P$ and $i,j \in \{1,\ldots,n\}$ with $i \ne j$, we write $S_{ij}$ for the twofold projection of $S$ into $A_i \times A_j$, that is, 
\[
S_{ij} := \{\langle a,b\rangle \in A_i \times A_j: \text{ there is } \bar{a} \in S \text{ such that } a_i = a \text{ and } a_j = b\}.
\]

\begin{lemma} \label{LEMA: C bueno no repite cruces}
Let $\mathcal{C}$ be a set of crosses of $P$, and put $S := \bigcap \mathcal{C}$. Suppose that for each $i,j \in \{1,\ldots,n\}$ with $i \ne j$ we have that either $S_{ij} = A_i \times A_j$ or $S_{ij}$ is a cross of $A_i \times A_j$. If $S \subseteq \cross{i}{a}{j}{b} \cap \cross{i}{a'}{j}{b'}$, then $a = a'$ and $b = b'$.
\end{lemma}

\begin{proof}
Suppose first that neither of the equalities $a = a'$, $b = b'$ holds. Then, it is easy to see that $S_{ij} \subseteq S'_{ij} = \{\langle a',b\rangle,\langle a,b'\rangle\}$, which implies that $S_{ij}$ is neither $A_i \times A_j$ nor a cross of $A_i \times A_j$, contradicting our hypotheses. Suppose now that exactly one of the equalities holds, say, $a = a'$. Then, $S_{ij} \subseteq S_{ij}' = \{a\} \times A_j$, which again is inconsistent with our hypotheses.
\end{proof}

Let $\mathcal{C}$ be a set of crosses of $P$. A cross $C \in \mathcal{C}$ is {\em $\mathcal{C}$-reducible} provided there are compatible $D,E \in \bar{\mathcal{C}}$ satisfying $C = D \circ E$. We say that $C$ is {\em $\mathcal{C}$-irreducible} if it is not $\mathcal{C}$-reducible.

Here is the main result of this subsection.

\begin{theorem} \label{TEO: descomposicion en cruces irreducibles}
Let $\mathcal{C}$ be a set of crosses of $P$, and put $S := \bigcap \mathcal{C}$. Suppose that for each $i,j \in \{1,\ldots,n\}$ with $i \ne j$ we have that either $S_{ij} = A_i \times A_j$ or $S_{ij}$ is a cross of $A_i \times A_j$.
\begin{enumerate}[$1.$]
\item For any cross $C$ of $P$ we have that $S \subseteq C$ iff $C \in \bar{\mathcal{C}}$.
\item Let $\mathcal{I}$ be the set of $\mathcal{C}$-irreducible crosses, and $\mathcal{B} \subseteq \bar{\mathcal{C}}$. Then, $\bar{\mathcal{B}} = \bar{\mathcal{C}}$ iff $\mathcal{I} \subseteq \mathcal{B}$. 
\end{enumerate}
\end{theorem}

\begin{proof}
We start by defining the directed graph $\mathbf{G} := \langle G^+ \cup G^-, E \cup F\rangle$ with
\begin{itemize}
\item $G^+ := \{ia : i \in \{1,\ldots,n\}, a \in A_i\}$,
\item $G^- := \{\neg ia : i \in \{1,\ldots,n\}, a \in A_i\}$,
\item $E := \{\langle ia,\neg ia'\rangle: i \in \{1,\ldots,n\}, a,a' \in A_i, a \ne a'\}$,
\item $F := \{\langle \neg ia,jb\rangle: \cross{i}{a}{j}{b} \in \mathcal{C}\}$.
\end{itemize}
We also set $G := G^+ \cup G^-$. Given $g,g' \in G$ we write $g \to g'$ if $\langle g,g'\rangle \in E \cup F$ and $g \path g'$ if $g = g'$ or there is a directed path from $g$ to $g'$ in $\mathbf{G}$. We also put $\neg g := \neg ia$ if $g = ia$ and $\neg g := ia$ if $g = \neg ia$. 

The proof of the theorem follows from a series of claims. We start with an easy observation.

\begin{claim} \label{CL: caminos contrarreciprocos}
If $g \path g'$, then $\neg g' \path \neg g$.
\end{claim}

For a subset $X \subseteq G$ let $X^\path := \{g \in G: \text{ there is } x \in X \text{ with } x \path g\}$. We say that $X$ is {\em inconsistent} if there is $g \in G$ such that $g, \neg g \in X^\path$; $X$ is {\em consistent} if it is not inconsistent.

Next for $\bar{a} \in P$ let $X_{\bar{a}} := \{ia_i: i \in \{1,\ldots,n\}\}$.

\begin{claim} \label{CL: asignaciones en S sii consistencia}
For all $\bar{a} \in P$ we have that $\bar{a} \in S$ if and only if $X_{\bar{a}}$ is consistent.
\end{claim}

For the forward implication let $Y := \{\neg ia': i \in \{1,\ldots,n\}, a' \in A_i \setminus \{a_i\}\}$. Note that $Y$ is the set of nodes that can be reached from $X_{\bar{a}}$ in exactly one step. Thus, $X_{\bar{a}}^\path = X_{\bar{a}} \cup Y^\path$. Note that $Y \subseteq G^-$, so if $\neg ia' \in Y$ and $\neg ia' \to g$, the edge $\langle \neg ia',g\rangle$ must belong to $F$. Hence, there is a cross $\cross{i}{a'}{j}{b} \in \mathcal{C}$ such that $g = jb$. Since $\bar{a} \in S \subseteq  \cross{i}{a'}{j}{b}$ and $a_i \ne a'$, we have that $a_j = b$, that is, $g \in X_{\bar{a}}$. This proves that $X_{\bar{a}}^\path = X_{\bar{a}} \cup Y$, so $X_{\bar{a}}$ is consistent.

For the converse assume $X_{\bar{a}}$ is consistent and let $\cross{i}{a}{j}{b} \in \mathcal{C}$. Suppose $a_i \ne a$. Then, we have a path $ia_i \to \neg ia \to jb$ in $\mathbf{G}$, and thus $a_j = b$ because otherwise $X_{\bar{a}}$ would be inconsistent. This proves that $\bar{a} \in \cross{i}{a}{j}{b}$. 

\begin{claim} \label{CL: los singuletes son consistentes}
For all $i \in \{1,\ldots,n\}$ and $a \in A_i$ we have that $\{ia\}$ is consistent.
\end{claim}

Note that, by hypothesis, the $i$-th canonical projection of $S$ is onto $A_i$. Thus, there is $\bar{a} \in S$ such that $a_i = a$, and by Claim \ref{CL: asignaciones en S sii consistencia}, $X_{\bar{a}}$ is consistent.

\begin{claim} \label{CL: caracterizacion de inconsistencia}
If $X \subseteq G^+$ is inconsistent, there are $x,y \in X$ with $x \ne y$ and $x \path \neg y$.
\end{claim}

Since $X$ is inconsistent, there is $g \in G$ such that $g, \neg g \in X^\path$, that is, there are $x,y \in X$ with $x \path g$ and $y \path \neg g$. By Claim \ref{CL: los singuletes son consistentes}, $x \ne y$, and by Claim \ref{CL: caminos contrarreciprocos}, $g \path \neg y$, and so $x \path \neg y$.

\begin{claim} \label{CL: consistentes se extienden a asignaciones completas de S}
If $X \subseteq G^+$ is consistent, there is $\bar{a} \in S$ such that $X \subseteq X_{\bar{a}}$. 
\end{claim}

Let $X \subseteq G^+$ be consistent, and suppose $Z$ is maximal among the consistent subsets of $G^+$ containing $X$. Observe that $Z = Z^\path \cap G^+$ by maximality. Let $I := \{i \in \{1,\ldots,n\}: ia \in Z \text{ for some } a \in A_i\}$; we claim that $I = \{1,\ldots,n\}$. For the sake of contradiction, suppose there is $j \in \{1,\ldots,n\} \setminus I$. Fix $b \in A_j$ and note that $Z \cup \{jb\}$ is inconsistent. By Claim \ref{CL: caracterizacion de inconsistencia}, there must be $ia \in Z$ such that $ia \path \neg jb$. Since the only edges incoming to $\neg jb$ are of the form $\langle jb',\neg jb\rangle$, there must be $b' \in A_j$ such that $ia \path jb'$. Hence, $jb' \in Z^\path \cap G^+ = Z$, a contradiction.

Finally, note that for each $i \in \{1,\ldots,n\}$ there is exactly one $a_i \in A_i$ such that $ia_i \in Z$, because otherwise $Z$ would be inconsistent. Thus, $Z = X_{\bar{a}}$, and by Claim \ref{CL: asignaciones en S sii consistencia} we have $\bar{a} \in S$. 

\medskip

For $\mathcal{B} \subseteq \bar{\mathcal{C}}$ define $F_\mathcal{B} := \{\langle \neg ia,jb\rangle: \cross{i}{a}{j}{b} \in \mathcal{B}\}$. Given $g,g' \in G$ we say that there is a {\em $\mathcal{B}$-path} from $g$ to $g'$, and write $g \overset{\mathcal{B}}{\path} g'$, if there is a path in $\mathbf{G}$ from $g$ to $g'$ with edges in $E \cup F_\mathcal{B}$.

\begin{claim} \label{CL: camino sii cruz}
Let $\mathcal{B} \subseteq \bar{\mathcal{C}}$,  $i,j \in \{1,\ldots,n\}$, $i \ne j$, and $a \in A_i$, $b \in A_j$. Then, $\neg ia \overset{\mathcal{B}}{\path} jb$ if and only if $\cross{i}{a}{j}{b} \in \bar{\mathcal{B}}$.
\end{claim}

The direction from right to left is a straightforward induction and therefore left to the reader. We prove the remaining implication. Suppose $\neg ia \overset{\mathcal{B}}{\path} jb$. We proceed by induction on the length of the $\mathcal{B}$-path from $\neg ia$ to $jb$. If this length is one, the claim follows from the definition of $\mathcal{B}$-path. So suppose there are $k \in \{1,\ldots,n\}$, with $k \ne i$, and $c,c' \in A_k$, with $c \ne c'$, such that $\neg ia \to kc \to \neg kc' \overset{\mathcal{B}}{\path} jb$ with $\langle \neg ia,kc\rangle \in F_\mathcal{B}$. Clearly, we can assume without loss that $kc \ne jb$. Observe that $k \ne j$. Otherwise, $jc \to \neg j c' \path jb \to \neg jc$, and $\{jc\}$ would be inconsistent in contradiction with Claim \ref{CL: los singuletes son consistentes}. Thus, by the induction hypothesis, the cross $\cross{k}{c'}{j}{b} \in \bar{\mathcal{B}}$. Furthermore, since $\langle \neg ia, kc\rangle \in F_\mathcal{B}$, we have $\cross{i}{a}{k}{c} \in \mathcal{B}$. Finally, as these crosses are compatible, $\cross{i}{a}{j}{b} = \cross{i}{a}{k}{c} \circ \cross{k}{c'}{j}{b} \in \bar{\mathcal{B}}$. 

\begin{claim}
If $C$ is a cross of $P$ such that $S \subseteq C$, then $C \in \bar{\mathcal{C}}$.
\end{claim}

Let $C := \cross{i}{a}{j}{b}$ be a cross of $P$ such that $S \subseteq C$. Take $a' \in A_i \setminus \{a\}$ and $b' \in A_j \setminus \{b\}$. Observe that, by Claim \ref{CL: consistentes se extienden a asignaciones completas de S}, the set $\{ia',jb'\}$ must be inconsistent. Now, by Claim \ref{CL: caracterizacion de inconsistencia}, we have $ia' \path \neg jb'$. Thus, there are $a'' \in A_i$ and $b'' \in A_j$ such that $ia' \to \neg ia'' \path jb'' \to \neg jb'$. Using Claim \ref{CL: camino sii cruz}, the cross $\cross{i}{a''}{j}{b''} \in \bar{\mathcal{C}}$. Hence, by Lemma \ref{LEMA: cruces - direccion facil}, we have that $S \subseteq \cross{i}{a''}{j}{b''}$, and thus $S \subseteq S' := \cross{i}{a}{j}{b} \cap \cross{i}{a''}{j}{b''}$. By Lemma \ref{LEMA: C bueno no repite cruces}, $a = a''$ and $b = b''$.

\medskip

This takes care of 1.\ since the above claim is the direction from left to right and the other direction follows from Lemma \ref{LEMA: cruces - direccion facil}. We complete the proof with the following three claims.

\begin{claim} 
There are no directed cycles in $\mathbf{G}$.
\end{claim}

For the sake of contradiction suppose there is a directed cycle from $ia$ to $ia$ for some $i \in \{1,\ldots,n\}$ and $a \in A_i$ (the case where the cycle starts and ends in $G^-$ is symmetrical in view of Claim \ref{CL: caminos contrarreciprocos}). By the definition of $\mathbf{G}$ there are $j \in \{1,\ldots,n\}$ with $j \ne i$, $a' \in A_i$ and $b,b' \in A_j$ with $a \ne a'$ and $b \ne b'$ such that $ia \to \neg ia' \to jb' \to \neg jb \path ia$. Thus, by Claim \ref{CL: camino sii cruz}, we have $\cross{i}{a}{j}{b}, \cross{i}{a'}{j}{b'} \in \bar{\mathcal{C}}$, and by 1.\, this entails $S \subseteq \cross{i}{a}{j}{b} \cap \cross{i}{a'}{j}{b'}$. This is a contradiction in view of Lemma \ref{LEMA: C bueno no repite cruces}.

\begin{claim} \label{CL: I barra = C barra}
$\bar{\mathcal{I}} = \bar{\mathcal{C}}$.
\end{claim}

The inclusion $\subseteq$ is obvious. Take $\cross{i}{a}{j}{b} \in \bar{\mathcal{C}}$, note that $\neg ia \path jb$. Since $\mathbf{G}$ is acyclic, we can take a path $\pi$ of maximum length from $\neg ia$ to $jb$. Note that $\pi$ must be an $\mathcal{I}$-path because otherwise $\pi$ could be lengthened. Now, by Claim \ref{CL: camino sii cruz}, we have that $\cross{i}{a}{j}{b} \in \bar{\mathcal{I}}$.

\begin{claim} \label{CL: D barra = C barra implica I < D}
If $\mathcal{B} \subseteq \bar{\mathcal{C}}$ is such that $\bar{\mathcal{B}} = \bar{\mathcal{C}}$, then $\mathcal{I} \subseteq \mathcal{B}$. 
\end{claim}

Take $C \in \mathcal{I}$. Note that $C \in \bar{\mathcal{B}}$ and, since $C$ is $\mathcal{C}$-irreducible, it cannot belong to $\bar{\mathcal{B}} \setminus \mathcal{B}$. 

\medskip

Finally, observe that the direction from left to right in 2.\ is Claim \ref{CL: D barra = C barra implica I < D} and the other direction follows from Claim \ref{CL: I barra = C barra}.
\end{proof}

\subsubsection{Congruence systems in $\mathcal{D}$}

Next we apply the combinatorial analysis in the previous section to investigate congruence systems in $\mathcal{D}$.

Let $\mathbf{A}$ be an algebra and let $M(x,y,z)$ be a ternary term in the language of $\mathbf{A}$. Recall that $M$ is a {\em majority term} for $\mathbf{A}$ provided that $\mathbf{A} \vDash M(x,x,y) = M(x,y,x) = M(y,x,x) = x$. Note also that any dual discriminator term for $\mathbf{A}$ is a majority term.

\begin{lemma}[\cite{BakerPixley-PolyInterpCRT}] \label{LEMA: m -> subalg determinadas por ij-projecciones}
Let $\mathbf{A}_1, \ldots, \mathbf{A}_n$, with $n \geq 2$, be algebras with a common majority term, and let $\mathbf{A}, \mathbf{B}$ be finite subalgebras of $\mathbf{A}_1 \times \ldots \times \mathbf{A}_n$. Then, if $A_{ij} \subseteq B_{ij}$ for all $i,j \in \{1,\ldots,n\}$ with $i \ne j$, we have $A \subseteq B$.
\end{lemma}

Let $\mathbf{A}$ be a subalgebra of a product $\mathbf{A}_1 \times \ldots \times \mathbf{A}_n$, and let $\pi_i|_A\colon A \to A_i$ be given by $\pi_i|_A(a_1,\ldots,a_n) := a_i$. The algebra $\mathbf{A}$ is a {\em subdirect subalgebra} (of $\mathbf{A}_1 \times \ldots \times \mathbf{A}_n$) provided that each $\pi_i|_A$ is onto. In this context, for $i \in \{1,\ldots,n\}$ we write $\rho_i$ to denote the kernel of $\pi_i|_A$, that is,
\[
\rho_i := \{\langle\bar{a},\bar{b}\rangle \in A^2: a_i = b_i\}.
\]
A subdirect algebra $\mathbf{A}$ of $\mathbf{A}_1 \times \ldots \times \mathbf{A}_n$ is {\em irredundant}, in symbols $\mathbf{A} \leq_{isd} \mathbf{A}_1 \times \ldots \times \mathbf{A}_n$, if for all $i,j \in \{1,\ldots,n\}$ we have that $\rho_i \subseteq \rho_j$ implies $i = j$.

\begin{lemma} \label{LEMA: subalg cuadrado - disc dual}
Let $\mathbf{A}$ and $\mathbf{B}$ be algebras with a common dual discriminator term. Let $\mathbf{S}\leq_{isd} \mathbf{A}\times\mathbf{B}$. One of the following holds:
\begin{enumerate}[$1.$]
\item $S = A \times B$.
\item \label{item:3} $S = (\{a\}\times B)\cup(A \times \{b\})$ for some $a\in A$ and $b\in B$.
\end{enumerate}
\end{lemma}

\begin{proof}
The result follows from \cite[Theorem 2.4]{FriedPixley-DualDiscrimnator} by observing that $S$ cannot be the graph of a homomorphism since $\mathbf{S}$ is irredundant.
\end{proof}

The combination of Lemmas \ref{LEMA: m -> subalg determinadas por ij-projecciones} and \ref{LEMA: subalg cuadrado - disc dual} provides a powerful tool to handle subproducts in dual discriminator varieties, which, together with Theorem \ref{TEO: descomposicion en cruces irreducibles}, allows us to fully unravel the relationship between subproducts and sets of crosses.

\begin{lemma} \label{LEMA: coro cruces}
Let $\mathbf{A} \leq_{isd} \mathbf{A}_1 \times \ldots \times \mathbf{A}_n$, and suppose there is a common dual discriminator term for $\mathbf{A}_1, \ldots, \mathbf{A}_n$. Let $\mathcal{C}_A$ be the set of crosses of $A_1 \times \ldots \times A_n$ containing $A$. The following holds:
\begin{enumerate}[$1.$]
\item $\bigcap \mathcal{C}_A = A$.

\item $\bar{\mathcal{C}}_A = \mathcal{C}_A$.

\item For $\mathcal{B} \subseteq \mathcal{C}_A$ the following are equivalent:
\begin{enumerate}[\rm (a)]
\item $\bigcap \mathcal{B} = A$.
\item $\bar{\mathcal{B}} = \mathcal{C}_A$.
\item Each $\mathcal{C}_A$-irreducible cross is in $\mathcal{B}$.
\end{enumerate}
\end{enumerate}
\end{lemma}

\begin{proof}
1. Put $S := \bigcap \mathcal{C}_A$; by Lemma \ref{LEMA: m -> subalg determinadas por ij-projecciones}, it suffices to show that $S_{ij} = A_{ij}$ for all $i,j$. Fix $i,j \in \{1,\ldots,n\}$ with $i \ne j$. If $A_{ij} = A_i \times A_j$, clearly $A_{ij} = S_{ij}$ since $A \subseteq S$. Otherwise, by Lemma \ref{LEMA: subalg cuadrado - disc dual}, there are $a \in A_i$ and $b \in A_j$ such that $A_{ij} = (\{a\} \times A_j) \cup (A_i \times \{b\})$. Thus, $\cross{i}{a}{j}{b} \in \mathcal{C}_A$ and we are done.

2. This follows from 1.\ and Lemma \ref{LEMA: cruces - direccion facil}. 

3. (a)$\Rightarrow$(b) Clearly $\bar{\mathcal{B}} \subseteq \bar{\mathcal{C}}_A = \mathcal{C}_A$. For the other inclusion, take a cross $C \in \mathcal{C}_A$. So $\bigcap \mathcal{B} = A \subseteq C$, and by item 1.\ of Theorem \ref{TEO: descomposicion en cruces irreducibles} with $\mathcal{C} = \mathcal{B}$ we have $C \in \bar{\mathcal{B}}$.

(b)$\Rightarrow$(a) This follows from 1.\ and Lemma \ref{LEMA: cruces - direccion facil}. 

(b)$\Leftrightarrow$(c) Since $\bar{\mathcal{C}}_A = \mathcal{C}_A$, the equivalence follows from item 2.\ of Theorem \ref{TEO: descomposicion en cruces irreducibles} with $\mathcal{C} = \mathcal{C}_A$.
\end{proof}

We now turn to apply the previous results to the study of congruence systems in $\mathcal{D}$. We need the following well-known fact about congruences of semisimple congruence-distributive algebras. We include a sketch of the proof for convenience of the reader.

\begin{lemma} \label{LEMA: congs <-> F's}
Let $\mathbf{A} \leq_{isd} \mathbf{A}_1 \times \ldots \times \mathbf{A}_n$ be a congruence distributive algebra and suppose $\mathbf{A}_1,\ldots,\mathbf{A}_n$ are simple. The map $$\theta \mapsto F_\theta := \{i \in \{1,\ldots,n\}: a_i = b_i \text{ for all } \langle a,b\rangle \in \theta\}$$ is a dual lattice-isomorphism from $\Con \mathbf{A}$ onto the power set of $\{1,\ldots,n\}$ ordered by inclusion, whose inverse is $$F \mapsto \theta_F := \{\langle a,b\rangle \in A^2: a_i = b_i \text{ for all } i \in F\}.$$
\end{lemma}

\begin{proof}
From \cite{FosterPixley64-CDimpliesFactCong} we know that, if $\theta$ is a congruence of $\mathbf{A} \leq_{sd} \mathbf{A}_1 \times \ldots \times \mathbf{A}_n$ and $\mathbf{A}$ is congruence distributive, then there are unique congruences $\theta_j$ of $\mathbf{A}_j$ for $j \in \{1,\ldots,n\}$ such that $\theta = \langle \theta_1 \times \ldots \times \theta_n\rangle \cap A^2$. (Here $\theta_1 \times \ldots \times \theta_n := \{\langle \bar{a},\bar{b}\rangle \in (A_1 \times \ldots \times A_n\rangle^2: \langle a_i,b_i\rangle \in \theta_i \text{ for } i \in \{1,\ldots,n\}\}$.) Now, as each $\mathbf{A}_j$ is simple, the only choices for $\theta_j$ are the diagonal and total congruences of $\mathbf{A}_j$. From here the proof of the lemma follows easily.
\end{proof}

The notation $\theta_F, F_\theta$ is used in the sequel whenever the hypotheses of Lemma \ref{LEMA: congs <-> F's} hold.

Given $A \subseteq A_1 \times \ldots \times A_n$ and $F \subseteq \{1,\ldots,n\}$ we define
\[
A_F := \{b \in A_1 \times \ldots \times A_n: \text{ there is } a \in A \text{ with } a|_F = b|_F\}.
\]
Also, for $\mathcal{F} \subseteq \mathcal{P}(\{1,\ldots,n\})$ let
\[
A_\mathcal{F} := \bigcap_{F \in \mathcal{F}} A_F.
\]
Note that $A \subseteq A_F$ for any $F \subseteq \{1,\ldots,n\}$, and, if $\mathbf{A}_1, \ldots, \mathbf{A}_n$ are algebras and $A$ is a subuniverse of $\mathbf{A}_1 \times \ldots \times \mathbf{A}_n$, then $A_F$ is a subuniverse as well.

Next we present a first characterization of CR tuples for algebras in $\mathcal{D}$ in terms of crosses.

\begin{lemma} \label{LEMA: CRT en coordenadas}
Let $\mathbf{A} \leq_{isd} \mathbf{A}_1 \times \ldots \times \mathbf{A}_n$, with $n \geq 2$, where $\mathbf{A}_1, \ldots, \mathbf{A}_n$ have a common dual discriminator term. Let $\mathcal{C}_A$ be the set of crosses of $A_1 \times \ldots \times A_n$ containing $A$. 
\begin{enumerate}
\item For all $F \subseteq \{1,\ldots,n\}$ and $C \in \mathcal{C}_A$ we have
\[
A_F \subseteq C \text{ iff both indices of $C$ belong to $F$.}
\]

\item Let $\theta_1, \ldots, \theta_k$ be congruences of $\mathbf{A}$ such that $\bigcap_{\ell = 1}^k \theta_\ell = \Delta$, and put $\mathcal{F} := \{F_{\theta_1}, \ldots, F_{\theta_k}\}$. The following are equivalent:
\begin{enumerate}[$(a)$]
\item $\langle \theta_1, \ldots, \theta_k\rangle$ is a CR tuple of $\mathbf{A}$,
\item $A_\mathcal{F} = A$.
\item For every $\mathcal{C}_A$-irreducible cross $C$, there is $F \in \mathcal{F}$ such that both indices of $C$ belong to $F$.
\end{enumerate}
\end{enumerate}
\end{lemma}

\begin{proof}
1. Fix $F \subseteq \{1,\ldots,n\}$ and $C := \cross{i}{a}{j}{b} \in \mathcal{C}_A$. The direction from right to left is immediate. For the other direction, we show that if $i \notin F$, a contradiction arises. Note that, since $n \geq 2$ and $\mathbf{A}$ is irredundant, each of $\mathbf{A}_1, \ldots, \mathbf{A}_n$ is nontrivial. In particular, since $\mathbf{A}$ is subdirect, there is $\bar{a} \in A$ with $a_j \ne b$. Let $a' \in A_i \setminus \{a\}$, and let $\bar{c} \in A_1 \times \ldots \times A_n$ be such that $c_i = a'$ and $c_t = a_t$ for $t \in \{1,\ldots,n\} \setminus \{i\}$. Note that $\bar{c} \in A_F \setminus C$, a contradiction.

2. To improve readability we write $F_\ell$ instead of $F_{\theta_\ell}$ for $\ell \in \{1,\ldots,n\}$.  Note that Lemma \ref{LEMA: congs <-> F's} says that $\bigcup \mathcal{F} = \{1,\ldots,n\}$. By the same lemma, for $\langle a_1,\ldots,a_n\rangle \in A^k$ we have that
\begin{align}
\langle \theta_1, & \ldots,\theta_k, a_1,\ldots,a_k\rangle \text{ is a system } \nonumber \\
& \text{ iff }  a_\ell|_{F_\ell \cap F_m} = a_m|_{F_\ell \cap F_m} \text{ for } \ell,m \in \{1,\ldots,k\}. \label{EQ: sistemas como conjuntos}
\end{align}

(a)$\Rightarrow$(b) Take $b \in A_\mathcal{F}$, and observe that, by definition of $A_\mathcal{F}$, there are $a_1, \ldots, a_k \in A$ such that $b|_{F_\ell} = a_\ell|_{F_\ell}$ for $\ell \in \{1,\ldots,k\}$. So $a_\ell|_{F_\ell \cap F_m} = b|_{F_\ell \cap F_m} = a_m|_{F_\ell \cap F_m}$ for $\ell, m \in \{1,\ldots,k\}$, and thus, $\langle \theta_1,\ldots,\theta_k,a_1,\ldots,a_k\rangle$ is a system on $\mathbf{A}$ by \eqref{EQ: sistemas como conjuntos}. By hypothesis, there is a solution $a \in A$ to this system, which, by Lemma \ref{LEMA: congs <-> F's}, implies that $a|_{F_\ell} = a_i|_{F_\ell}$ for $\ell \in \{1,\ldots,k\}$. Since $\bigcup \mathcal{F} = \{1,\ldots,n\}$, this says that $b = a \in A$.

(b)$\Rightarrow$(a) Suppose $\langle \theta_1,\ldots,\theta_k,a_1,\ldots,a_k\rangle$ is a system on $\mathbf{A}$. By \eqref{EQ: sistemas como conjuntos} there is (a unique) $b \in A_1 \times \ldots \times A_n$ such that $b|_{F_\ell} = a_\ell|_{F_\ell}$ for $\ell \in \{1,\ldots,k\}$. Hence, $b \in A_\mathcal{F} \subseteq A$. Using \eqref{EQ: sistemas como conjuntos} again we see that $b$ is a solution to $\langle \theta_1,\ldots,\theta_k,a_1,\ldots,a_k\rangle$.

To prove the equivalence between (b) and (c) we first establish the following.

\begin{claim} \label{CL: int B = A_F}
Let
\[
\mathcal{B} := \{C \in \mathcal{C}_A: \text{there is $F \in \mathcal{F}$ such that both indices of $C$ belong to $F$}\}.
\]
We have that $\bigcap \mathcal{B} = A_\mathcal{F}$.
\end{claim}

Fix $F \in \mathcal{F}$, and let $\mathcal{B}_F := \{C \in \mathcal{C}_A: \text{both indices of $C$ belong to $F$}\}$. Item 1.\ says that $\mathcal{B}_F$ is the set of crosses of $A_1 \times \ldots \times A_n$ that contain $A_F$. Since $\mathbf{A}$ is an irredundant subdirect product of the $\mathbf{A}_i$'s, so is any extension of $\mathbf{A}$ in $\mathbf{A}_1 \times \ldots \times \mathbf{A}_n$, in particular $\mathbf{A}_F$. Thus, by Lemma \ref{LEMA: coro cruces} applied to $\mathbf{A}_F$, we have $\bigcap \mathcal{B}_F = A_F$. Hence, since $\bigcup_{F \in \mathcal{F}} \mathcal{B}_F = \mathcal{B}$, we have $\bigcap \mathcal{B} = A_\mathcal{F}$.

Finally, in view of Claim \ref{CL: int B = A_F}, the equivalence (b)$\Leftrightarrow$(c) follows at once from the equivalence (a)$\Leftrightarrow$(c) of Lemma \ref{LEMA: coro cruces}.(3).
\end{proof}

The final stretch to Theorem \ref{TEO: CRT tuple for dual disc - equivalences} is achieved by translating condition $(c)$ in Lemma \ref{LEMA: CRT en coordenadas} into a purely congruential property.

\begin{lemma} \label{LEMA: bridge lemma}
Let $\mathbf{A} \leq_{isd} \mathbf{A}_1 \times \ldots \times \mathbf{A}_n$, and suppose there is a common dual discriminator term for $\mathbf{A}_1, \ldots, \mathbf{A}_n$. Let $\mathcal{C}_A$ be the set of crosses of $A_1 \times \ldots \times A_n$ containing $A$. The following holds for $i,j \in \{1,\ldots,n\}$ with $i \ne j$:
\begin{enumerate}[$1.$]
\item There is a cross in $\mathcal{C}_A$ with indices $i,j$ iff $\rho_i$ and $\rho_j$ do not permute.
\item Suppose $C \in \mathcal{C}_A$ is a cross with indices $i,j$, then $C$ is $\mathcal{C}_A$-reducible iff there is $k \in \{1,\ldots,n\} \setminus \{i,j\}$ such that $\rho_k \subseteq \rho_i \circ \rho_j$.
\end{enumerate}
\end{lemma}

\begin{proof}
1. Note that, since $\rho_i$ and $\rho_j$ are maximal congruences of $\mathbf{A}$, we have that $\rho_i \circ \rho_j = \rho_j \circ \rho_i$ iff $\rho_i \circ \rho_j = \nabla$. Now, it is easy to see that $\rho_i \circ \rho_j = \nabla$ iff $A_{ij} = A_i \times A_j$. So, by Lemma \ref{LEMA: subalg cuadrado - disc dual}, we are done.

2. The implication from left to right is an easy exercise; we prove the remaining implication. Suppose there is $k \in \{1,\ldots,n\} \setminus \{i,j\}$ such that $\rho_k \subseteq \rho_i \circ \rho_j$. Note that $\rho_i \circ \rho_k \subseteq \rho_i \circ \rho_j$; thus $\rho_i$ cannot permute with $\rho_k$ because otherwise $\rho_i \circ \rho_k = \rho_i \circ \rho_j = \nabla$, and this is contradicts the fact that $A$ is contained in $C$. Analogously, $\rho_k$ and $\rho_j$ do not permute, and by 1.\ there are crosses $\cross{i}{a}{k}{c}$ and $\cross{k}{c'}{j}{b}$ in $\mathcal{C}_A$ for some $a \in A_i$, $b \in A_j$, $c,c' \in A_k$. We claim that if $c = c'$, then $A_{ij} = A_i \times A_j$. In fact, suppose $c = c'$ and take $x \in A_i$, $y \in A_j$; so $\langle x,c\rangle \in A_{ik}$ and $\langle c,y\rangle \in A_{kj}$. Now, from the fact that $\rho_k \subseteq \rho_i \circ \rho_j$ it follow that $\langle x,y\rangle \in A_{ij}$, finishing the proof of the claim. Since by assumption $A_{ij}$ is a cross, we must have $c \ne c'$. Hence, $\cross{i}{a}{k}{c}$ and $\cross{k}{c'}{j}{b}$ are compatible, and their composition  $\cross{i}{a}{j}{b}$ is in $\bar{\mathcal{C}}_A$ and thus in $\mathcal{C}_A$ by 2.\ from Lemma \ref{LEMA: coro cruces}. Finally, by Lemma \ref{LEMA: C bueno no repite cruces}, we have that $C = \cross{i}{a}{j}{b}$ is $\mathcal{C}_A$-reducible.
\end{proof}

We are now in a position to present the main result of this section. 

\begin{theorem} \label{TEO: CRT tuple for dual disc - equivalences}
Let $\mathbf{A}$ be an algebra in $\mathcal{D}$, and let $\theta_1,\ldots, \theta_k$ be congruences of $\mathbf{A}$. Let $\Sigma$ be the set of meet irreducible congruences of $\mathbf{A}$ that contain $\bigcap _{j=1}^k \theta_j$. The following are equivalent:
\begin{enumerate}

\item $\langle \theta_1, \ldots, \theta_k\rangle$ is a CR tuple of $\mathbf{A}$,

\item for all $\lambda,\mu \in \Sigma$ one of the following conditions holds:
\begin{itemize}
\item $\lambda \circ \mu = \mu \circ \lambda$,
\item there is $\sigma \in \Sigma \setminus \{\lambda ,\mu\}$ with $\sigma \subseteq \lambda \circ \mu$,
\item there is $j \in \{1,\ldots,k\}$ such that $\theta_j \subseteq \lambda \cap \mu$.
\end{itemize}
\end{enumerate}
\end{theorem}

\begin{proof}
By Lemma \ref{LEMA: CRT <-> CRT diagonal} and the Correspondence Theorem it suffices to show that the theorem holds under the additional assumption $\bigcap_{\ell = 1}^k \theta_\ell = \Delta$. By Birkhoff's subdirect decomposition theorem we may assume without loss that $\mathbf{A} \leq_{isd} \mathbf{A}_1 \times \ldots \times \mathbf{A}_n$ with $\Sigma = \{\rho_1,\ldots,\rho_n\}$. (The irredundancy follows from the fact that each $\rho_i$ is a maximal congruence of $\mathbf{A}$.) Let $\mathcal{F} := \{F_{\theta_1}, \ldots, F_{\theta_k}\}$, and let $\mathcal{C}_A$ be the set of crosses of $A_1 \times \ldots \times A_n$ containing $A$. Then, by Lemma \ref{LEMA: CRT en coordenadas}.(2), item (1) is equivalent to
\begin{enumerate}
\item[$(2')$] For every $\mathcal{C}_A$-irreducible cross $C$, there is $F \in \mathcal{F}$ such that both indices of $C$ belong to $F$.
\end{enumerate}
Finally, $(2')$ easily translates to (2) using Lemma \ref{LEMA: bridge lemma}.
\end{proof}

Since finite distributive lattice are included in $\mathcal{D}$, Theorem \ref{TEO: CRT tuple for dual disc - equivalences} immediately yields a characterization of CR tuples for these structures. An entertaining exercise is to directly derive Corollary \ref{CORO: caract CR tuple for DIST LAT} from Theorem \ref{TEO: CRT tuple for dual disc - equivalences}.

\subsubsection{\prob{\sf CRT$|_\mathcal{D}$} is poly-time computable}

We now employ Theorem \ref{TEO: CRT tuple for dual disc - equivalences} to characterize the complexity of \prob{CRT$|_\mathcal{D}$}. First, a useful lemma.

\begin{lemma} \label{LEMA: MI cong in poly-time}
There is a poly-time algorithm that, given a finite congruence-distributive algebra $\mathbf{A}$ as input, returns the set of meet-irreducible congruences of $\mathbf{A}$.
\end{lemma}

\begin{proof}
We start by observing that the set $J$ of join-irreducible congruences
of $\mathbf{A}$ can be computed in poly-time. Recall that given $a,b\in A$ the {\em principal congruence} $\mathbf{\theta}^{\mathbf{A}}(a,b)$ is the last congruence of $\mathbf{A}$ containing $\langle a,b\rangle$. We know from \cite{Bergman02-ComputComplexCongruences} that $\theta^\mathbf{A}(a,b)$ is poly-time
computable from $\mathbf{A}$. Of course, each join-irreducible congruence
of $\mathbf{A}$ is principal, thus to compute $J$, we first compute
the set $P:=\{\theta^{\mathbf{A}}(a,b):a,b\in A\}$, and then check
for each $\theta\in P$ whether it is join-irreducible via the test
$\theta>\bigvee\{\delta\in P:\delta<\theta\}$. Next, given $\theta\in J$
let $\theta':=\bigvee\{\delta\in J:\delta\nleq\theta\}$; clearly
each $\theta'$ is computable in poly-time. Finally, since $\con\mathbf{A}$
is a distributive lattice, we have that $\{\theta':\theta\in J\}$
is exactly the set of meet-irreducible congruences of $\mathbf{A}$.
\end{proof}

\begin{theorem} \label{TEO: CRT.D is in P}
\prob{CRT$|_\mathcal{D}$} is in \textup{P}.
\end{theorem}

\begin{proof}
Let $\mathbf{A}$, $\theta_1, \ldots, \theta_k$ be an instance of \prob{CRT$|_\mathcal{D}$}. By Lemma \ref{LEMA: MI cong in poly-time} the set $\Sigma$ of meet irreducible congruences of $\mathbf{A}$ that contain $\bigcap _{j=1}^k \theta_j$ is computable in poly-time.  Finally, observe that the conditions stated in 2.\ of Theorem \ref{TEO: CRT tuple for dual disc - equivalences} are easily poly-time computable once $\Sigma$ is known.
\end{proof}

\section{(Almost) Dichotomy for Varieties Generated by Two-Element Algebras}
\label{SEC: Almost Dichotomy}

In this final section we study the complexity of \prob{CRT$|_\mathcal{V}$} for  an arbitrary variety $\mathcal{V}$ generated by a two-element algebra. This is possible in part due to a result by Post \cite{Post41-PostLattice}, that shows that (up to a term-equivalence) there are countable many of these varieties and provides concrete descriptions for each of them. We fall slightly short of establishing a dichotomy, as we are able to show that each of the cases is either tractable or \textup{coNP}-complete, with the exception of the case of semilattices, which is left as an open problem. As we shall see, the heavy lifting to obtain the results below is already done, since they follow with little extra effort from the content of the previous sections leveraged by Post's classification and the concept of interpretation between varieties.

Let $\mathcal{V}, \mathcal{V}'$ be varieties in languages $L, L'$, respectively. An {\em interpretation} of $\mathcal{V}$ in $\mathcal{V}'$ is a mapping $T$ from $L$ into $L'$-terms satisfying:
\begin{itemize}
\item if $c$ in $L$ is a constant symbol, then $T(c)$ is a unary $L'$-term such that $\mathcal{V}' \vDash T(c)(x) \approx T(c)(y)$;
\item if $f$ in $L$ is a function symbol of arity $n > 0$, then $T(f)$ is an $n$-ary $L'$-term;
\item for every $\mathbf{A} \in \mathcal{V}'$ the algebra $\mathbf{A}^T := \langle A, T(s)^\mathbf{A}\rangle_{s \in L}$ is in $\mathcal{V}$. (For constant $c \in L$ the operation $T(c)^\mathbf{A}$ is the nullary operation on $A$ corresponding to the value of $T(c)$ in $\mathbf{A}$).
\end{itemize}
The varieties $\mathcal{V}$ and $\mathcal{V}'$ are said to be {\em term-equivalent} if there are interpretations $T$ of $\mathcal{V}$ in $\mathcal{V}'$ and $S$ of $\mathcal{V}'$ in $\mathcal{V}$ such that $(\mathbf{A}^T)^S = \mathbf{A}$ for all $\mathbf{A} \in \mathcal{V}'$ and $(\mathbf{A}^S)^T = \mathbf{A}$ for all $\mathbf{A} \in \mathcal{V}$.

If $\mathcal{V}$ is a variety in the language $L$ and $c$ is a constant symbol not in $L$, we write $\mathcal{V}(c)$ to denote the variety in the language $L \cup \{c\}$ where no restrictions are imposed on the new constant.

Given computational problems \prob{X} and \prob{Y} we write \prob{X} $\leq_P$ \prob{Y} whenever there is a poly-time many-one reduction from \prob{X} to \prob{Y}.

The relevance of interpretations for our task lies in the fact that they instantly produce polynomial reductions.

\begin{lemma} \label{LEM: interpret -> reduccion}
Let $\mathcal{V}$ and $\mathcal{V}'$ be varieties.
\begin{enumerate}[\rm (a)]
\item If there is an interpretation of $\mathcal{V}$ in $\mathcal{V}'$, then \prob{CRT$|_{\mathcal{V}'}$} $\leq_P$ \prob{CRT$|_{\mathcal{V}}$}.
\item \prob{CRT$|_{\mathcal{V}}$} $\leq_P$ \prob{CRT$|_{\mathcal{V}(c)}$}. 
\end{enumerate}
\end{lemma}

\begin{proof}
(a). If $T$ is an interpretation of $\mathcal{V}$ in $\mathcal{V}'$, then the map $\langle \mathbf{A},\theta_1,\ldots,\theta_k\rangle \mapsto \langle \mathbf{A}^T,\theta_1,\ldots,\theta_k\rangle$ is a polynomial reduction from \prob{CRT$|_{\mathcal{V}'}$} to \prob{CRT$|_{\mathcal{V}}$}.

(b). The map that takes $\langle \mathbf{A},\theta_1,\ldots,\theta_k\rangle$ to $\langle \langle \mathbf{A},a\rangle,\theta_1,\ldots,\theta_k\rangle$, where $a$ is arbitrarily chosen in $A$, is a polynomial reduction from \prob{CRT$|_{\mathcal{V}}$} to \prob{CRT$|_{\mathcal{V}(c)}$}.
\end{proof}

Next we introduce the varieties that perform as the main actors in our analysis, by specifying two-element algebras generating them.  Let $2 := \{0,1\}$. We write $\neg, \wedge, \vee, +$ to denote the usual operations on $2$, and define the following ternary operations on 2:
\begin{itemize}
\item $\mathsf{s}(x,y,z) := x + y + z$,
\item $\mathsf{n}(x,y,z) := (x \wedge y) \vee z$,
\item $\mathsf{m}(x,y,z) := (x \wedge y) \vee (x \wedge z) \vee (y \wedge z)$.
\end{itemize}
Given an algebra $\mathbf{A}$, let $V(\mathbf{A})$ denote the variety generated by $\mathbf{A}$. Define:
\begin{itemize}
\item $\mathcal{S} := V(2)$,
\item $\mathcal{U} := V(\langle 2, \neg, 0, 1\rangle)$,
\item $\mathcal{A} := V(\langle 2, \mathsf{s}\rangle)$,
\item $\mathcal{N} := V(\langle 2, \mathsf{n}\rangle)$,
\item $\mathcal{M} := V(\langle 2, \mathsf{m}\rangle)$,
\item $\mathcal{J} := V(\langle 2, \vee\rangle)$, $\mathcal{J}_0 := V(\langle 2, \vee, 0\rangle)$, $\mathcal{J}_1 := V(\langle 2, \vee, 1\rangle)$, $\mathcal{J}_{01} := V(\langle 2, \vee,0,1\rangle)$.
\end{itemize}
Observe that $\mathcal{S}$ is the class of all sets.

The choice of this list is explained by the fact that every variety generated by a two-element algebra is either interpretable in or interprets one of the above defined varieties.

Recall that a term-operation of an algebra $\mathbf{A}$ is just the interpretation in $A$ of a term in the language of $\mathbf{A}$. 

\begin{lemma} \label{LEM: Post}
Let $\mathcal{V}$ be a variety generated by the two-element algebra $\mathbf{A}$. Then, one of the following holds:
\begin{enumerate}[\rm (a)]
\item $\mathsf{s}(x,y,z) := x + y + z$ is a term-operation of $\mathbf{A}$, and thus, $\mathcal{A}$ is interpretable in $\mathcal{V}$;
\item $\mathsf{n}(x,y,z) := (x \wedge y) \vee z$ is a term-operation of $\mathbf{A}$, and thus, $\mathcal{N}$ is interpretable in $\mathcal{V}$;
\item $\mathsf{m}(x,y,z) :=(x \wedge y) \vee (x \wedge z) \vee (y \wedge z)$ is a term-operation of $\mathbf{A}$, and thus, $\mathcal{M}$ is interpretable in $\mathcal{V}$;
\item $\mathcal{V}$ is interpretable in $\mathcal{U}$;
\item $\mathcal{V}$ is term-equivalent to $\mathcal{J}$, $\mathcal{J}_0$, $\mathcal{J}_1$ or $\mathcal{J}_{01}$.
\end{enumerate}
\end{lemma}

\begin{proof}
Given an algebra $\mathbf{A}$ let $\Clo \mathbf{A}$ denote the set of all term-operations of positive arity of $\mathbf{A}$.
Given algebras $\mathbf{A},\mathbf{B}$ with universe $\{0,1\}$ we write $\mathbf{A} \sqsubseteq \mathbf{B}$ provided that $\Clo \mathbf{A} \subseteq \Clo \mathbf{B}$. The relation $\sqsubseteq$ is a preorder whose associated equivalence relation we denote by $\equiv$. It is not hard to see that $\mathbf{A} \equiv \mathbf{B}$ iff $V(\mathbf{A})$ and $V(\mathbf{B})$ are term-equivalent. The equivalence classes of $\equiv$ ordered by $\sqsubseteq$ form a countable lattice, known as Post's lattice, which is completely described in \cite{Post41-PostLattice}. The lemma follows from the above observations and this description.
\end{proof}

We already know that \prob{CRT$|_\mathcal{N}$} is in \textup{P} (Theorem \ref{TEO: CRT.N is in P}); in the next series of lemmas we consider the remaining varieties in our list.

\begin{lemma} \label{LEM: triple suma en P}
\prob{CRT$|_{\mathcal{A}}$} is in \textup{P}.
\end{lemma}

\begin{proof}
Using the shorthand $x +_w y := \mathsf{s}(x,y,w)$, note that $\mathcal{A}$ satisfies the identities:
\begin{itemize}
\item $(x +_w y) +_w z \approx x +_w (y +_w z)$,
\item $x +_w y \approx y +_w x$,
\item $x +_w w \approx x$,
\item $x +_w x \approx w$.
\end{itemize}
Hence, given $\mathbf{A} \in \mathcal{A}$ and $e \in A$ the algebra $\mathbf{A}_e := \langle A, +_e, -, e, \lambda_0, \lambda_1\rangle$, where $-x := x$, $\lambda_0(x) := e$ and $\lambda_1(x) := x$, is a $\mathbb{Z}_2$-vector space. It follows that there is an interpretation of the variety of $\mathbb{Z}_2$-vector spaces in $\mathcal{A}(c)$, where is a new constant symbol. So, by Lemma \ref{LEM: interpret -> reduccion} and Theorem \ref{TEO: vector spaces is in P} we are done.
\end{proof}

\begin{lemma} \label{LEM: U01 es hard}
\prob{CRT$|_{\mathcal{U}}$} is \textup{coNP}-complete.
\end{lemma}

\begin{proof}
By the proof of Theorem \ref{TEO: hardness para sets} it suffices to show that \prob{CRT$|_\mathcal{S}$} $\leq_P$ \prob{CRT$|_{\mathcal{U}}$}. 

Given a set $A$ and any element $c \in A$ define the algebra $\mathbf{A}_c := \langle A \cup A', \neg, c, c' \rangle$ where
\begin{itemize}
\item $A' := \{a': a \in A\}$ is a set disjoint with $A$ such that $|A'| = |A|$;
\item $\neg a := a'$ and $\neg a' := a$ for all $a \in A$.
\end{itemize}
Since $\mathcal{U}$ is axiomatized by the identities $\neg \neg x \approx x$ and $\neg 0 \approx 1$, it is clear that $\mathbf{A}_c \in \mathcal{U}$. Also, given an equivalent relation $\theta$ on $A$ let $\tilde{\theta} := \theta_i \cup \{\langle a',b'\rangle: \langle a,b\rangle \in \theta\}$. We leave it to the reader to check that the transformation taking an instance $\langle A,\theta_1,\ldots,\theta_k\rangle$ of \prob{CRT$|_\mathcal{S}$} and an arbitrary $c \in A$ to $\langle \mathbf{A}_c,\tilde{\theta}_1,\ldots,\tilde{\theta}_k\rangle$ is in fact a polynomial reduction from \prob{CRT$|_{\mathcal{S}}$} to \prob{CRT$|_{\mathcal{U}}$}.
\end{proof}

\begin{proposition}
The restrictions of \prob{CRT} to the varieties $\mathcal{J}$, $\mathcal{J}_0$, $\mathcal{J}_1$, $\mathcal{J}_{01}$ are all polynomially equivalent.
\end{proposition}

\begin{proof}
Taking into account the obvious interpretaions among these four varieties, by Lemma \ref{LEM: interpret -> reduccion}.(a), we have that \prob{CRT$|_{\mathcal{J}_{01}}$} $\leq_P$ \prob{CRT$|_{\mathcal{J}_0}$} $\leq_P$ \prob{CRT$|_{\mathcal{J}}$} and \prob{CRT$|_{\mathcal{J}_{01}}$} $\leq_P$ \prob{CRT$|_{\mathcal{J}_1}$} $\leq_P$ \prob{CRT$|_{\mathcal{J}}$}. Thus, it suffices to show that \prob{CRT$|_{\mathcal{J}}$} $\leq_P$ \prob{CRT$|_{\mathcal{J}_{01}}$}. Consider the transformation taking an instance $\langle \mathbf{A},\theta_1,\ldots,\theta_k\rangle$ to $\langle \mathbf{A}^0,\theta_1^0,\ldots,\theta_k^0\rangle$, where $\mathbf{A}^0$ is the bounded semilattice obtained from $\mathbf{A}$ by adding a new bottom element $0$ and interpreting $1$ as the top element of $\mathbf{A}$, and $\theta_i^0 := \theta_i \cup \{\langle 0,0\rangle\}$ for $i \in \{1,\ldots,k\}$. We leave it to the reader to check that this is in fact a polynomial reduction from \prob{CRT$|_{\mathcal{J}}$} to \prob{CRT$|_{\mathcal{J}_{01}}$}.
\end{proof}

Sadly, we were not able to characterize the complexity of \prob{CRT$|_\mathcal{J}$}, and thus we propose the following.

\begin{problem*}
What is the complexity of \prob{CRT$|_\mathcal{J}$}?
\end{problem*}

With that out of the way, we present now our (almost) dichotomy theorem.

\begin{theorem}
Let $\mathcal{V}$ be a variety generated by the two-element algebra $\mathbf{A}$. Suppose $\mathcal{V}$ is not term-equivalent to any of $\mathcal{J}$, $\mathcal{J}_0$, $\mathcal{J}_1$, $\mathcal{J}_{01}$. Then, either
\begin{itemize}
\item $\mathsf{s}$, $\mathsf{n}$ or $\mathsf{m}$ is a term-operation of $\mathbf{A}$, and \prob{CRT$|_\mathcal{V}$} is in \textup{P}; or
\item $\mathcal{V}$ is interpretable in $\mathcal{U}$, and \prob{CRT$|_\mathcal{V}$} is \textup{coNP}-complete.
\end{itemize}
\end{theorem}

\begin{proof}
If $\mathsf{s}$, $\mathsf{n}$ or $\mathsf{m}$ is a term-operation of $\mathbf{A}$, then $\mathcal{V}$ satisfies condition (a), (b) or (c) in Lemma \ref{LEM: Post}. So, it follows from Lemma \ref{LEM: interpret -> reduccion}.(a) together with Lemma \ref{LEM: triple suma en P}, Theorem \ref{TEO: CRT.N is in P}, and Theorem \ref{TEO: CRT.D is in P} that \prob{CRT$|_\mathcal{V}$} is in \textup{P}. (For (c) note that $\mathsf{m}$ is the dual discriminator function on 2.)

Otherwise, by Lemma \ref{LEM: Post} and our assumption on $\mathcal{V}$, we have that $\mathcal{V}$ is intepretable in $\mathcal{U}$, so Lemmas \ref{LEM: interpret -> reduccion}.(a) and \ref{LEM: U01 es hard} say that \prob{CRT$|_\mathcal{V}$} is \textup{coNP}-complete.
\end{proof}

\bibliographystyle{abbrv}

\end{document}